\documentclass[a4paper,UKenglish,cleveref, autoref, thm-restate, nolineno]{lipics-v2021}

\newboolean{cameraready}
\newboolean{long}
\newcommand{\AppendixSymbol}{\ding{72}}
\setboolean{cameraready}{false} 
\setboolean{long}{true} %
\ifthenelse{\boolean{cameraready}}{\setboolean{long}{false}}{}
\ifthenelse{\boolean{long}}{
	\NewDocumentEnvironment{prooflater}{m}{\begin{proof}}{\end{proof}\ignorespacesafterend}
	\NewDocumentEnvironment{proofsketch}{o +b}{}{\ignorespacesafterend}
	\newcommand{\restateref}[1]{}
	\NewDocumentEnvironment{statelater}{m}{}{}

	\NewDocumentCommand{\onlyShort}{+m}{}
	\NewDocumentCommand{\onlyLong}{+m}{#1}
	\NewDocumentCommand{\shortLong}{+m +m}{#2}
}{
	\NewDocumentEnvironment{prooflater}{m +b}{%
		\newcounter{#1-usages}\setcounter{#1-usages}{0}%
		\expandafter\global\expandafter\def\csname#1\endcsname{\stepcounter{#1-usages}\begin{proof}#2\end{proof}}%
		\AtEndDocument{\ifnumequal{\value{#1-usages}}{0}{
            \ifthenelse{\boolean{cameraready}}{}{\todo[inline]{use prooflater #1}}
        
        }{}}%
	}{\ignorespacesafterend}
	\NewDocumentEnvironment{proofsketch}{O{Proof sketch.}}{\begin{proof}[#1]}{\end{proof}\ignorespacesafterend}
	\usepackage{apptools}

        \ifthenelse{\boolean{cameraready}}{
	   \newcommand{\restateref}[1]{
            [\IfAppendix{\AppendixSymbol{}}{\AppendixSymbol{}}]
          }
        }{
            \newcommand{\restateref}[1]{[\IfAppendix{\hyperref[#1]{\AppendixSymbol{}}}{\hyperref[#1*]{\AppendixSymbol{}}}]}
        }

	\NewDocumentEnvironment{statelater}{m +b}{%
		\newcounter{#1-usages}\setcounter{#1-usages}{0}%
		\expandafter\global\expandafter\def\csname#1\endcsname{\stepcounter{#1-usages}#2}%
		\AtEndDocument{\ifnumequal{\value{#1-usages}}{0}{
        \ifthenelse{\boolean{cameraready}}{}{\todo[inline]{use statelater #1
        }}}{}}%
	}{\ignorespacesafterend}

	\NewDocumentCommand{\onlyShort}{+m}{#1}
	\NewDocumentCommand{\onlyLong}{+m}{}
	\NewDocumentCommand{\shortLong}{+m +m}{#1}
}

\pdfoutput=1 %
\ifthenelse{\boolean{cameraready}}{}{\hideLIPIcs}
\hideLIPIcs  %

\graphicspath{{./graphics/}}%

\title{Linear Layouts Revisited: \\
Stacks, Queues, and Exact Algorithms}
\titlerunning{Linear Layouts Revisited: Stacks, Queues, and Exact Algorithms}

\newcommand\blfootnote[1]{%
  \begingroup
  \renewcommand\thefootnote{}\footnote{#1}%
  \addtocounter{footnote}{-1}%
  \endgroup
}

\author{Thomas Depian}{Algorithms and Complexity Group, TU Wien, Vienna, Austria}{tdepian@ac.tuwien.ac.at}{https://orcid.org/0009-0003-7498-6271}{Project No. 10.47379/ICT22029 of the Vienna Science Foundation (WWTF).}

\author{Simon D.~Fink}{Algorithms and Complexity Group, TU Wien, Vienna, Austria}{sfink@ac.tuwien.ac.at}{https://orcid.org/0000-0002-2754-1195}{Project No. 10.47379/ICT22029 of the Vienna Science Foundation (WWTF) and Project No. 10.55776/Y1329 of the Austrian Science Fund (FWF).}

\author{Robert Ganian}{Algorithms and Complexity Group, TU Wien, Vienna, Austria}{rganian@gmail.com}{https://orcid.org/0000-0002-7762-8045}{Project No. 10.47379/ICT22029 of the Vienna Science Foundation (WWTF) and Projects No. 10.55776/Y1329, 10.55776/COE12
of the Austrian Science Fund (FWF).}

\author{Vaishali Surianarayanan}{University of California at Santa Barbara, Santa Barbara, CA 93106 USA}{vaishali@ucsb.edu}{https://orcid.org/0000-0003-3091-3823}{}

\authorrunning{T. Depian, S\,D. Fink, R. Ganian and V. Surianarayanan} %

\Copyright{Thomas Depian, Simon D. Fink, Robert Ganian and Vaishali Surianarayanan} %



\category{} %

\relatedversion{} %
\nolinenumbers %

\EventEditors{Anne Benoit, Haim Kaplan, Sebastian Wild, and Grzegorz Herman}
\EventNoEds{4}
\EventLongTitle{33rd Annual European Symposium on Algorithms (ESA 2025)}
\EventShortTitle{ESA 2025}
\EventAcronym{ESA}
\EventYear{2025}
\EventDate{September 15--17, 2025}
\EventLocation{Warsaw, Poland}
\EventLogo{}
\SeriesVolume{351}
\ArticleNo{13}
\usepackage{todonotes}
\usepackage{booktabs}
\usepackage{complexity}
\usepackage{pifont}
\usepackage{subcaption}
\usepackage{amsmath}
\usepackage{mathtools}
\usepackage{fancybox}
\usepackage{mdframed}
\usepackage{cite}
\usepackage{pifont}
\usepackage{xspace}
\usepackage{amssymb}
\usepackage{amsthm}

	\let\oldrestatable\restatable
	\def\restatable{\expandafter\oldrestatable}

\NewDocumentCommand{\linLs}{o o}{\ensuremath{\langle\prec^+\IfNoValueF{#1}{_{#1}},\sigma^+\IfNoValueF{#1}{_{#1}}\rangle\IfNoValueF{#2}{_{#2}}}\xspace}

\NewDocumentCommand{\linL}{o o}{\ensuremath{\langle\prec\IfNoValueF{#1}{_{#1}},\sigma\IfNoValueF{#1}{_{#1}}\rangle\IfNoValueF{#2}{_{#2}}}\xspace}
\NewDocumentCommand{\stackL}{o}{\ensuremath{\linL[#1][\mathcal{S}]}\xspace}
\NewDocumentCommand{\queueL}{o}{\ensuremath{\linL[#1][\mathcal{Q}]}\xspace}
\NewDocumentCommand{\pw}{o}{\ensuremath{\omega\IfNoValueF{#1}{(#1)}}\xspace}
\NewDocumentCommand{\LeftOfV}{o m}{\ensuremath{\text{Left}\IfNoValueF{#1}{^{#1}}_{#2}}\xspace}
\NewDocumentCommand{\LeftAtOfV}{o m}{\ensuremath{\LeftOfV[#1]{#2} \cup \{#2\}}\xspace}
\NewDocumentCommand{\RightOfV}{o m}{\ensuremath{\text{Right}\IfNoValueF{#1}{^{#1}}_{#2}}\xspace}
\NewDocumentCommand{\EdgesCrossingV}{o m}{\ensuremath{X\IfNoValueF{#1}{^{#1}}_{#2}}\xspace}
\NewDocumentCommand{\EdgesCrossingVPage}{o o m m}{\ensuremath{\EdgesCrossingV[#1]{#3} \cap \sigma\IfNoValueF{#2}{_{#2}}^{-1}(#4)}\xspace}
\NewDocumentCommand{\DirectedCut}{o m}{\ensuremath{(#2, \IfNoValueTF{#1}{\prec_{#2}}{#1})}\xspace}

\NewDocumentCommand{\StateProcessed}{m}{\ensuremath{L(#1)}\xspace}
\NewDocumentCommand{\StateNotProcessed}{m}{\ensuremath{R(#1)}\xspace}

\newcommand{\probname}[1]{{\normalfont\textsc{#1}}}

\newcommand{\bigoh}{\mathcal{O}}
\newcommand{\vi}{\textsf{\textup{vi}}}
\newcommand{\tw}{\textsf{\textup{tw}}}

\newcommand{\LabelOrdindary}{\ensuremath{\uparrow}\xspace}
\newcommand{\LabelArch}{\reflectbox{\rotatebox[origin=c]{180}{\ensuremath{\curvearrowright}}}\xspace}

\newcommand{\LabelImage}{\ensuremath{\{\LabelOrdindary, \LabelArch\}%
}\xspace}

\newcommand{\Size}[1]{\ensuremath{\left\vert #1 \right\vert}}
\newcommand{\BigO}[1]{\ensuremath{\mathcal{O}(#1)}}

\newcommand{\oldtodo}[2][noinline]{}

\crefname{AC}{$\mathcal{A}$C}{$\mathcal{A}$Cs}
\Crefname{AC}{$\mathcal{A}$C}{$\mathcal{A}$Cs}
\crefname{NC}{$\mathcal{N}$C}{$\mathcal{N}$Cs}
\Crefname{NC}{$\mathcal{N}$C}{$\mathcal{N}$Cs}

\begin{document}

\maketitle

\begin{abstract}
In spite of the extensive study of stack and queue layouts, many fundamental questions remain open concerning the complexity-theoretic frontiers for computing stack and queue layouts. A stack (resp.\ queue) layout places vertices along a line and assigns edges to pages so that no two edges on the same page are crossing (resp.\ nested).
We provide three new algorithms which together substantially expand our understanding of these problems:

\begin{enumerate}
\item A fixed-parameter algorithm for computing minimum-page stack and queue layouts w.r.t.\ the vertex integrity of an $n$-vertex graph $G$. This result is motivated by an open question in the literature and generalizes the previous algorithms parameterizing by the vertex cover number of $G$. The proof relies on a newly developed Ramsey pruning technique. Vertex integrity intuitively measures the vertex deletion distance to a subgraph with only small connected components.

\item An $n^{\bigoh(q \ell)}$ algorithm for computing $\ell$-page stack and queue layouts of page width at most $q$. This is the first algorithm avoiding a double-exponential dependency on the parameters. The page width of a layout measures the maximum number of edges one needs to cross on any page to reach the outer face. 

\item A $2^{\bigoh(n)}$ algorithm for computing $1$-page queue layouts. This improves upon the previously fastest $n^{\bigoh(n)}$ algorithm and can be seen as a counterpart to the recent subexponential algorithm for computing $2$-page stack layouts [ICALP'24], %
but relies on an entirely different technique.
\end{enumerate}
\end{abstract}

\section{Introduction}
\label{sec:introduction}
A linear layout \linL of a graph $G$ is a total ordering $\prec$ of its vertices and a partitioning $\sigma$ of its edges into $k$ \emph{pages} such that each page satisfies certain conditions. The two by far most commonly studied types of linear layouts are \emph{stack layouts} and \emph{queue layouts}; in the former, we require that no four vertices $a\prec b\prec c\prec d$ have the edges $ac$ and $bd$ placed on the same page, while for the latter we forbid any page from containing the edges $ad$ and $bc$. Intuitively, this corresponds to forbidding edges crossing and edges nesting (in a rainbow pattern), respectively---see \Cref{fig:example}. %
\onlyShort{\blfootnote{Due to space constraints, we defer full proofs of statements marked with \AppendixSymbol\ to the \ifthenelse{\boolean{cameraready}}{full version~\cite{ARXIV}}{appendix}.}}

\begin{figure}[h]
    \begin{center}
    \includegraphics[page=1]{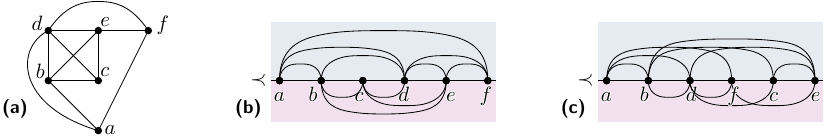}
    \end{center}
    \caption{A graph $G$ \textbf{\textsf{(a)}} together with a $2$-page stack layout \textbf{\textsf{(b)}} and $2$-page queue layout \textbf{\textsf{(c)}}. The two pages are colored blue and lilac, respectively.}
    \label{fig:example}
\end{figure}

Originally motivated from applications in VLSI~\cite{chung1987embedding,rosenberg1986book} and bioinformatics~\cite{haslinger1999rna}, stack and queue layouts have become the focus of extensive research. While there are many classical works studying the structural properties of the two notions~\cite{heath1987embedding,yannakakis1989embedding,heath1992laying,dujmovic2004linear,dujmovic2005layout,dujmovic2007graph,dujmovic2011book,dujmovic2020planar}, several more recent studies have targeted exact and parameterized algorithms for computing stack and queue layouts that minimize the number of used pages. These results are inherently aimed at circumventing the long-known intractability of these problems: the \NP-hardness of determining whether an input graph admits a $2$-page stack layout or a $1$-page queue layout has been established over thirty years ago~\cite{yannakakis1989embedding,heath1992laying}.

To that end, recent works have pursued the use of structural graph parameters to establish tractability on graphs which are ``well-structured''. Computing $\ell$-page stack layouts is now known to be fixed-parameter tractable w.r.t.\ the vertex cover number~\cite{bhore2020parameterized} or, alternatively, the feedback edge number of the input graph $G$~\cite{GanianMOPR24}. Computing $\ell$-page queue layouts is likewise known to admit a fixed-parameter algorithm w.r.t.\ the vertex cover number of $G$~\cite{bhore2022parameterized}. One drawback of these results is that they only yield tractability under highly restrictive graph parameters, in the sense of achieving low values only on very ``simple'' graphs. While computing $1$-page queue layouts is also known to be fixed-parameter tractable when parameterized by the treedepth of $G$~\cite{bhore2022parameterized}, for general $\ell$ it was repeatedly posed as an open question whether the aforementioned vertex-cover based algorithms can be lifted to less restrictive structural graph parameters~\cite{dujmovic2011book,bhore2020parameterized,ganian2021parameterized,bhore2022parameterized,GanianMOPR24}. 
Another recent direction is the establishment of tighter algorithmic upper bounds for special cases: while $\ell$-page stack and queue layouts can both be computed in time $n^{\bigoh(n)}$ via trivial brute-force algorithms, $2$-page stack layouts are now known to admit a subexponential algorithm~\cite{GanianMOPR24}.

\subparagraph{Contributions.}
As our \textbf{first contribution}, we lift the aforementioned fixed-parameter tractability of computing linear layouts to a less restrictive structural graph parameter:

\begin{theorem}
\label{thm:vifpt}
Computing a minimum-page stack layout and minimum-page queue layout is fixed-parameter tractable w.r.t.\ the \emph{vertex integrity} of the input graph.
\end{theorem}

The vertex integrity $\vi(G)$ of a graph $G$ measures, roughly speaking, how many vertex deletions are required to decompose $G$ into small connected components, and has been used in the design of parameterized algorithms for a variety of challenging problems~\cite{dvovrak2017solving,bodlaender2020subgraph,ganian2021structural,ganian2021structuralpaths,dvovrak2021complexity,gima2022exploring,lampis2024fine,hanaka2024parameterized}.\footnote{Some prior works use the \emph{fracture number}, which is parametrically equivalent to vertex integrity.}
More precisely, $\vi(G)$ is the minimum integer $k$ satisfying the following: there exists a vertex set $X\subseteq V(G)$ such that for each connected component $C$ of $G-X$, $|V(C)\cup X|\leq k$. As vertex integrity is upper-bounded by the vertex cover number plus one, Theorem~\ref{thm:vifpt} generalizes the fixed-parameter tractable algorithms for computing stack and queue layouts w.r.t.\ the vertex cover number by Bhore, Ganian, Montecchiani, and Nöllenburg~\cite{bhore2020parameterized,bhore2022parameterized}.
Moreover, since the vertex integrity is sandwiched between the vertex cover number and treedepth, our result can be seen as a stepping stone towards solving the problem on more general parameterizations.

To establish Theorem~\ref{thm:vifpt}, we introduce a novel proof technique we call \emph{Ramsey pruning}. On a high level, the algorithm underlying the computation is very simple: assuming $\Size{G}$ is larger than some pre-designated function $k$, it recursively identifies a connected component $C$ of $G-X$ such that $G-X$ contains sufficiently many copies of $C$, and deletes $C$ from the instance. The difficulty lies in proving that this operation is safe, i.e., that $G-C$ does not admit a layout with fewer pages than $G$. The ``usual'' approach for such a proof would be to show that $C$ can be reinserted into any hypothetical linear layout \linL of $G-C$ without increasing the number of pages; however, this not only seems excruciatingly difficult to prove, but we also believe it to be false for certain choices for \linL.
Ramsey pruning avoids this issue by using Ramsey-type arguments to argue that \linL must contain a \emph{guiding sub-layout} $\left\langle\prec',\sigma'\right\rangle$---a linear layout of some carefully selected subgraph of $G-C$---with certain well-defined properties. We then use the well-formedness of $\left\langle\prec',\sigma'\right\rangle$ to build a brand new linear layout $\left\langle\prec^*,\sigma^*\right\rangle$ of $G$, thus establishing that $G$ and $G-C$ indeed require the same number of pages. We believe this new technique to be generic and applicable to other problems under the same parameterization, as suggested by the fact that unlike the previous vertex-cover based algorithms our approach works for both stack and queue layouts with almost no problem-specific changes required. %

While Theorem~\ref{thm:vifpt} pushes the frontiers of tractability for our two problems of interest in a complexity-theoretic sense, our use of Ramsey-type arguments means that the obtained running time bound has an entirely impractical dependency on the parameter (see Lemma~\ref{lem:kernel}). In fact, up to now none of the known parameterized algorithms for computing $\ell$-page stack or queue layouts have a runtime parameter dependency that would be better than double-exponential---the parameterized algorithms w.r.t.\ the vertex cover~\cite{bhore2020parameterized,bhore2022parameterized} and feedback edge numbers~\cite{GanianMOPR24} all preprocess the graph to obtain an equivalent instance (a \emph{kernel}) whose size is exponential in the parameter, and then solve that equivalent bounded-size instance via brute force. As our \textbf{second contribution}, we provide the first single-exponential algorithms (in the parameters) capable of solving these problems for arbitrary fixed choices of $\ell$:

\begin{restatable}\restateref{thm:pagewidth}{theorem}{thmpagewidth}
\label{thm:pagewidth}
Given integers $\ell$ and $q$ along with an $n$-vertex graph $G$, we can compute an $\ell$-page stack or queue layout with page width at most $q$ of $G$ (if one exists) in time $n^{\bigoh(q\cdot\ell)}$.
\end{restatable}

Here, the page width measures the maximum number of edges one needs to cross to reach any vertex from the outer face (see also \Cref{sec:preliminaries}). The page width of linear layouts\footnote{The term \emph{cutwidth} has been used interchangeably with page width in the past~\cite{chung1987embedding}; here, we use the latter in order to disambiguate from the related graph parameter cutwidth.} has been studied in several early works~\cite{chung1987embedding,heath1987embedding,stohr1988trade,stohr1991pagewidth}
and it is generally desirable to obtain layouts not only using few pages, but also with low page width---in fact, the original papers introducing stack layouts explicitly targeted algorithms optimizing both measures.
However, the algorithmic applications of page width as a parameter have only been investigated in a recent paper on extending incomplete linear layouts~\cite{depian2024parameterized}.
We remark that while Theorem~\ref{thm:pagewidth} merely provides a so-called \emph{\XP\ algorithm} in a complexity-theoretic sense, it may still be more efficient in practice than known fixed-parameter algorithms for the problem---especially if the aim is to compute layouts with low pagewidth. 

A natural approach towards establishing Theorem~\ref{thm:pagewidth} would be to use dynamic programming in order to construct the sought-after linear layout in a ``left-to-right'' fashion; however, the issue is that doing this directly would require us to store roughly $2^n$ possible subsets of previously processed vertices. To circumvent this, we obtain new insights into the decomposability of layouts with bounded page width, allowing us to construct an $n^{\bigoh(q\cdot\ell)}$-size auxiliary \emph{state graph} $H$, where we show that the dynamic computation of a sought-after layout can be represented as an easily computable path in $H$. 

As our final \textbf{third contribution}, we take a step back from the multivariate analysis of these problems and recall that the trivial $n^{\bigoh(n)}$ barrier was only recently overcome for computing $2$-page stack layouts~\cite{GanianMOPR24}, i.e., the lowest-page stack layouts giving rise to an \NP-complete problem. Here, we show that the trivial $n^{\bigoh(n)}$ barrier can also be overcome when aiming for the lowest-page queue layouts which still give rise to \NP-completeness:

\begin{restatable}\restateref{thm:singleexp}{theorem}{thmsingleexp}
\label{thm:singleexp}
Given an $n$-vertex graph $G$, we can compute a $1$-page queue layout of $G$ (if one exists) in time $2^{\bigoh(n)}$.
\end{restatable}

We remark that the running time bound provided by Theorem~\ref{thm:singleexp} is worse than the $2^{\bigoh(\sqrt{n})}$ bound previously obtained for $2$-page stack layouts~\cite{GanianMOPR24}. The reason for this is that the latter result relied on the equivalence of $2$-page stack layouts with the existence of \emph{subhamiltonian paths}; the authors of that previous work then essentially obtained a single-exponential fixed-parameter algorithm for finding such paths w.r.t.\ the graph parameter treewidth, improving on the previously established fixed-parameter tractability of the problem~\cite{bannister2018crossing}. 

However, emulating the same approach seems difficult in the queue layout setting: $1$-page queue layouts also have an equivalent formulation in terms of so-called \emph{arched-level planar} drawings~\cite{heath1992laying}, but it is not at all clear whether computing such drawings is fixed-parameter tractable w.r.t.\ treewidth (not to mention the fact that one would need a single-exponential algorithm).
Instead, our proof of Theorem~\ref{thm:singleexp} relies on a Turing reduction from computing arched-level planar drawings to $2^{\bigoh(n)}$ many instances of a well-studied problem called \probname{Level Planarity}, which is known to be solvable in polynomial time~\cite{heath1996recognizing,junger1998level}.

\subparagraph{Related Work.} 
We refer to the dedicated survey for an overview of many earlier structural results concerning linear layouts~\cite{dujmovic2004linear}.
Key structural results in the area include the existence of $4$-page stack layouts~\cite{yannakakis1989embedding} 
and $42$-page queue layouts~\cite{dujmovic2020planar,bekos2023improved}
for all planar graphs.
Researchers have also studied the notion of \emph{mixed layouts}, which are linear layouts with some pages behaving like stack and some like queue layouts~\cite{angelini2022mixed,katheder_et_al:LIPIcs.STACS.2025.56}.
We note that due to the techniques used in their proofs, it seems very likely that Theorems~\ref{thm:vifpt} and~\ref{thm:pagewidth} could be adapted to the mixed layout setting with minimum changes required.
Finally, we remark that parameterized algorithms for stack and/or queue layouts have additionally been considered in the extension setting (where the task is to complete a provided partial layout)~\cite{depian2024parameterized,DFGN.PEQ.2025},
in the upward-planarity setting (where the graph is directed and edges must be oriented in a left-to-right fashion along the layout)~\cite{bhore2021upward} 
and when the vertex ordering in the layout is fixed~\cite{liu2021parameterized,agrawal2024eliminating}.

\section{Preliminaries}
\label{sec:preliminaries}

For an integer $\ell \geq 1$ we let $[\ell]$ denote the set $\{1, 2, \ldots, \ell\}$.
We assume the reader to be familiar with standard graph terminology~\cite{Die.GT4.2012}.
Without loss of generality, we assume all input graphs to be connected, have $n$ vertices and $m$ edges.
For a set of vertices $V' \in V(G)$, we let $G[V']$ denote the graph induced on $V'$.
Furthermore, for a set of edges $E' \in E(G)$, we let $V(E')$ denote the set of its endpoints and $G[E'] = (V(E'), E')$.
A \emph{cut} $(A, B)$ of $G$ is a partition of the vertices in $A \subseteq V(G)$ and $B = V(G) \setminus A$.
We call $F = \{uv \in E(G) \mid u \in A, v \in B\}$ a \emph{cut-set} of \emph{size} $\Size{F}$ that \emph{induces} the cut $(A, B)$.
Throughout the paper, we omit an explicit reference to $G$ if it is clear from the context, e.g., we write $V$ and $E$ instead of $V(G)$ and $E(G)$.

Given a linear order $\prec$ of a graph $G$ and two vertices $u, v \in V$, we say that $u$ is \emph{left} of $v$ if $u \prec v$ and \emph{right} of $v$ if $v \prec u$.
The vertices $u$ and $v$ are \emph{consecutive on the spine} if they occur consecutively in $\prec$, i.e., there is no vertex $w \in V$ such that $u \prec w \prec v$ or $v \prec w \prec u$ holds.
We address with $\LeftOfV[\prec]{v} \coloneqq \{u \in V \mid u \prec v\}$ all vertices $u \in V$ left of $v$ with respect to $\prec$ and define $\RightOfV[\prec]{v} \coloneqq V \setminus (\LeftAtOfV[\prec]{v})$.
Finally, we let $\EdgesCrossingV[\prec]{v} \coloneqq \{uw \in E \mid u \in \LeftAtOfV{v}, w \in \RightOfV{v}\}$ denote the set of edges that \emph{span} the spine between $v$ and its right neighbor, i.e., all edges with one endpoint left of (or at) $v$ and one right of $v$.

We assume familiarity with the basic foundations of parameterized complexity theory~\cite{CFK+.PA.2015}. All algorithms obtained in this work are exact, deterministic, constructive and rely exclusively on computable functions.
To express some of our bounds, we will occasionally use the Knuth notation $\uparrow\uparrow$ where for an integer $z$, $2\uparrow\uparrow z$ represents an exponential tower of $2$'s of height $z$.

\subparagraph*{Linear Layouts.}
Let $u_1v_1$ and $u_2v_2$ be two edges of $G$.
We say that they \emph{cross} under a linear order $\prec$ if $u_1 \prec u_2 \prec v_1 \prec v_2$ holds, and they \emph{nest} if $u_1 \prec u_2 \prec v_2 \prec v_1$ holds.
For an integer $\ell\geq 1$, let $\sigma_G \colon E(G) \to [\ell]$ be a function that assigns each edge to a \emph{page} $p \in [\ell]$.
The \emph{linear layout} \linL[G] is a \emph{stack \textup{(}queue\textup{)} layout} if no two edges $e_1, e_2 \in E$ with $\sigma_G(e_1) = \sigma_G(e_2)$ cross (nest).
We call~$\prec_G$ the \emph{spine order} and~$\sigma_G$ the \emph{page assignment}.
For the remainder of the paper, we write $\prec$ and $\sigma$ if the graph~$G$ is clear from context.
\onlyLong{We use \stackL[G] and \queueL[G] to differentiate between stack and queue layouts, respectively. }
\oldtodo{$\prec_{G[X]}$ is pD not induced -> change def}

The \emph{page width} \pw[\linL] of an $\ell$-page linear layout \linL corresponds to the maximum number of edges on a single page that span the spine between a vertex and its right neighbor, i.e., 
	$\pw[\linL] \coloneqq \max_{p \in [\ell]}\max_{v \in V(G)}\Size{\EdgesCrossingVPage{v}{p}}$.
We call an $\ell$-page linear layout with page width $q$ an \emph{$\ell$-page $q$-width linear layout}, or simply a \emph{solution} when the $\ell$ and $q$ are clear from context. In line with the literature, we assume no explicit bound on the page width when $q$ is not specified.

\subparagraph*{Vertex Integrity.} A graph $G$ has \emph{vertex integrity} $\vi(G) = p$ if $p$ is the smallest integer with the following property: $G$ contains a vertex set $S$ such that for each connected component $H$ of $G - S$, $|V(H) \cup S| \leq p$. One may observe that the vertex integrity is upper-bounded by the size of a minimum vertex cover in the graph (i.e., the vertex cover number) plus one. The vertex integrity of an $n$-vertex graph along with a corresponding partition into $S$ and $\mathcal{C}=G-S$ can be computed in time $\mathcal{O}(p^{p+1} \cdot n)$~\cite{DrangeDH16}.

\section{A Fixed-Parameter Algorithm Parameterized by Vertex Integrity}
\label{sec:fpt-vi}

In this section, we obtain a fixed-parameter algorithm that takes as input a positive integer $\ell$ and a graph $G$, is parameterized by $\ell+\vi(G)$ and computes an $\ell$-page stack and/or  queue layout of $G$ (or both), if such layouts exist. We begin by noting that using well-known relationships between vertex integrity, minimum-page stack and queue layouts and the graph parameter \emph{treewidth} ($\tw$), we can assume $\ell$ to be upper-bounded by a function of $\vi(G)$ thanks to the well-known relation $\tw(G) \leq \vi(G)$:

\begin{proposition}[{[\citen{dujmovic2007graph,Wiechert17}]}]
	The number of pages in a minimum stack and queue layout of a graph $G$ is upper-bounded by $\tw(G)+1\leq \vi(G)+1$ and $2^{\tw(G)} +1\leq 2^{\vi(G)}+1$, respectively.
\end{proposition}
The algorithm operates by obtaining a problem kernel, as formalized below; recall $p = \vi(G)$.

\begin{lemma}
	\label{lem:kernel}
	There is an $\mathcal{O}(p^{p+1} \cdot n)$ time algorithm that takes a graph $G$ with an integer $\ell$ and outputs a subgraph $G'$ of $G$ (a \emph{kernel}) of size at most $\big(2\uparrow\uparrow (p\cdot 2^{2p^2})\cdot 2 + 6\big)^{\ell \cdot p}$ with the following property: $G'$ admits an $\ell$-page stack (or queue) layout $\linL[G']$ if and only if so does $G$. Moreover, such a layout for $G$ can be computed from $\linL[G']$ in polynomial time.
\end{lemma}

Lemma~\ref{lem:kernel} directly implies Theorem~\ref{thm:vifpt}, and the rest of this section is dedicated to proving this lemma. For the following, let us fix a graph $G$, positive integers $\ell$ and $p$, and a choice of $S \subseteq V(G)$ witnessing $\vi(G)=p$ as in \Cref{sec:preliminaries}. Further let $\mathcal{C}$ be the set of connected components of $G-S$. First, we define a notion of ``component-types'' which groups components in $\mathcal{C}$ that exhibit the same outside connections and internal structure.

\begin{definition} We say two graphs $H_0, H_1 \in \mathcal{C}$ are \emph{twins}, denoted $H_0 \sim H_1$, if there exists a canonical isomorphism $\alpha$ from $H_0$ to $H_1$ such that for each vertex $u \in V(H_0)$ and each $v \in S$, $uv \in E(G)$ if and only if $\alpha(u)v \in E(G)$. 
\end{definition}

\begin{restatable}\restateref{obs:size-equiv-class}{lemma}{observationRamseySizeEquivClass}
\label{obs:size-equiv-class}
Each graph $H\in \mathcal{C}$ has at most $p$ vertices, $\sim$ is an equivalence relation and the number of equivalence classes in $[\sim]$ is upper-bounded by $p \cdot 2^{2p^2}$. Moreover, a partition of connected components into $[\sim]$ can be computed in time at most $\mathcal{O}(p \cdot 2^{2p^2}\cdot n)$.
\end{restatable}
\begin{prooflater}{pobservationRamseySizeEquivClass}
By definition of $\mathcal{C}$, each $H \in \mathcal{C}$ has at most $p$ vertices. So the number of non-isomorphic graphs can be upper-bounded by $p \cdot 2^{p^2}$. Since $|S| \leq p$, there are at most $p^2$ possible edges between $S$ and each $H \in \mathcal{C}$. Hence, we have $[\sim] \leq p \cdot 2^{2p^2}$. For the running time, it suffices to process the connected components of $G-S$ in an arbitrary order and use exhaustive branching over $\alpha$ to determine (in time at most $\mathcal{O}(p \cdot 2^{2p^2})$) which of the equivalence classes in $[\sim]$ it belongs to.
\end{prooflater}

We now introduce the notion of large equivalence classes based on their size. We then use this to define what we call a large group of vertices---one that contains exactly one representative
from each large equivalence class. Our kernel will keep a bounded number of these large groups.
\begin{definition}
\label{defn:large-eq-class}
Let $k$ be a positive integer, an equivalence class $[H]$ of $\sim$ is said to be {\em $k$-large} if $|[H]|\geq k$. Further a vertex set $L \subseteq V(G)$ is called a {\em $k$-large group} if the induced subgraph $G[L]$ is a disjoint union of exactly one graph from each $k$-large equivalence class of~$\sim$.
\end{definition}

Next we define a special induced subgraph that will serve as our kernel. The definition is based on carefully choosing a $k$ that is bounded by a computable function of $p$ and $\ell$ so that keeping only $k$ many $k$-large groups along with the small parts of the graph suffices to capture the necessary structure. 
Towards this, let us first fix $f(\ell,p,x):=2^{2^{\ell\cdot x^2\cdot 2^{12p^2}}}$ to be a computable function that is large enough to apply our Ramsey-type arguments later on; here, $x$ will be an integer that represents the size of a deletion set (initially $S$, but this will be updated iteratively in the proof of Lemma~\ref{lem:reducedcomputation}). Moreover, let $g(\ell,p)$ be a computable function of $\ell$ and $p$ that will upper-bound our nested application of the function $f$ in that same proof; to provide a concrete bound, we set $g(\ell, p):= \big(2\uparrow\uparrow (p\cdot 2^{2p^2})\cdot 2 + 4\big)^{\ell \cdot p}$.

\begin{definition}
\label{def:red-graph}
An induced subgraph $G'$ of $G$ is said to be a {\em reduced graph} of $G$ if there exists a positive integer $x\leq g(\ell,p)$ and a partition of $V(G') = S \uplus S' \uplus Y$ satisfying:
\begin{enumerate}
	\item $|S\cup S'|= x$ and \( S' \) is the set of all vertices in graphs in \( \mathcal{C} \) that do not belong to an $f(\ell,p,x)$-large equivalence class of $\sim$.
	\item If there are no $f(\ell,p,x)$-large equivalence classes of $\sim$, $Y=\emptyset$ and $V(G')=S\uplus S' = V(G)$. Otherwise, it holds that $Y=L_1 \uplus \cdots \uplus L_{f(\ell,p,x)}$ where \( L_i \) 
	is an $f(\ell,p,x)$-large group for each $i\in\{1,\cdots,f(\ell,p,x)\}$, and each pair of $L_i$ and $L_j$, $i \neq j$, is vertex-disjoint.
\end{enumerate}
\end{definition}

Intuitively, the reduced graphs we will be dealing with will consist of $S$, a set $S'$ of all equivalence classes of $\sim$ which are too small to fully saturate our large groups (these will later be treated essentially in the same way as $S$), and a sufficient number of large groups; equivalence classes of size larger than $f(\ell,p,x)$ are ``pruned'' to have size exactly $f(\ell,p,x)$. A schematic overview of this intuition can be found in Figures~\ref{fig:ramsey-overview}b and~c later on.

The core of our result is the following lemma, which we prove separately in \Cref{subsec:VI-Kernel-proof}.
\begin{lemma}
\label{lem:VI-kernel-proof}
If $G'$ is a reduced graph of $G$ then $G$ has an $\ell$-page stack (queue) layout if and only if $G'$ has an $\ell$-page stack (queue) layout.%
\end{lemma}

Before proceeding to that proof, we show how to construct a reduced graph having size bounded by a computable function of $p$ and $\ell$ in polynomial time.

\begin{restatable}\restateref{lem:reducedcomputation}{lemma}{lemmaReducedComputation}
\label{lem:reducedcomputation}
There exists a reduced graph $G'$ of $G$, and given $S$ and $\sim$ such a graph can be computed in polynomial time. Further, the number of vertices in $G'$ is upper-bounded by $\big(2\uparrow\uparrow (p\cdot 2^{2p^2})\cdot 2 + 6\big)^{\ell \cdot p}$.
\end{restatable}
\begin{prooflater}{plemmaReducedComputation}
We present an algorithm to compute a reduced graph $G'$ below:
\begin{enumerate}
	\item Initialize $X=S$ and $\mathcal{C}':=\mathcal{C}$
	\item As long as there exists an equivalence class $[H]$, $H\in \mathcal{C}'$ that is not $f(\ell,p,|X|)$-large:
	\begin{enumerate}
		\item Set $\mathcal{C}'=\mathcal{C}'\setminus [H]$ and $X=X\cup \{v:v\in H', H'\in [H]\}$
	\end{enumerate}
	\item If $\mathcal{C}'\neq \emptyset$, then set $O:=X\uplus L_1\uplus \cdots \uplus L_{f(\ell,p,|X|)}$ where each $L_{1\leq i\leq f(\ell,p,|X|)}$ is a $f(\ell,p,|X|)$-large group. Else set $O:=X$.
	\item Output $G'=G[O]$
\end{enumerate}

We first bound the size of $V(G')=O$. In each iteration of step~$2$, the vertices of all graphs in exactly one equivalence class $[H]$, $H\in \mathcal{C'}$ are added to the set $X$ and the graphs in $[H]$ are removed from $\mathcal{C'}$. Let $X_i$, $i\in \{1,\cdots,2^{2p^2}\}$ be the set $X$ at the end of the $i^{th}$ iteration of step $2$ and let there be $t$ iterations of step $2$. Note that $t=0$ if there is no successful iteration of step 2. Since there are at most $p\cdot 2^{2p^2}$ equivalence classes in $\sim$, the algorithm repeats step~2 at most $p\cdot 2^{2p^2}$ times, thus $t\leq p\cdot 2^{2p^2}$.

Let $X_0:=S$, observe that $|X_0|=|S|\leq p$. If $t\geq 1$, by construction and by the fact that each graph in $\mathcal{C}$ has at most $p$ vertices, for each $i\in \{1,\cdots,t\}$, $|X_i|\leq |X_{i-1}|+p\cdot f(\ell,p,|X_i|)$. Since $t\leq p\cdot 2^{2p^2}$, $|X_t|\leq g(\ell,p)\leq \big(2\uparrow\uparrow (p\cdot 2^{2p^2})\cdot 2 + 4\big)^{l\cdot p}$.

Further each $f(\ell,p,|X_t|)$-large group has at most $p\cdot2^{2p^2}$ vertices and thus $f(\ell,p,|X_t|)$ many $f(\ell,p,|X_t|)$-large groups have at most $p\cdot 2^{2p^2}\cdot f(\ell,p,|X_t|)$ vertices. Therefore we have $|O|\leq |X_t|+p\cdot 2^{2p^2}\cdot f(\ell,p,|X_t|)\leq g(\ell,p)+p\cdot 2^{2p^2}\cdot f(\ell,p,g(\ell,p)) \leq \big(2\uparrow\uparrow (p\cdot 2^{2p^2})\cdot 2 + 6\big)^{l\cdot p}$.

Next observe that when the algorithm stops all vertices in graphs in $\mathcal{C}$ that do not belong to an $f(\ell,p,|X_t|)$-large equivalence class belongs to $X_t$. Further $S\subseteq X_t$. For the $f(\ell,p,|X_t|)$-large equivalence classes, $f(\ell,p,|X_t|)$ many $f(\ell,p,|X_t|)$-large groups are added to $O$. Thus by construction, the induced graph $G'=G[O]$ output by the algorithm is a reduced subgraph by Definition~\ref{def:red-graph}. This completes the proof.
\end{prooflater}

Lemma~\ref{lem:reducedcomputation} combined with Lemma~\ref{lem:VI-kernel-proof}
establishes Lemma~\ref{lem:kernel} by constructing a reduced graph which is a kernel of the desired size.
Thus, what remains is to establish Lemma~\ref{lem:VI-kernel-proof}; this is where the core ideas of Ramsey pruning will come into play.

\subsection{Proof of Lemma~\ref{lem:VI-kernel-proof}}
\label{subsec:VI-Kernel-proof}
For this proof, we fix $G'$ to be a reduced graph of $G$.
Since $G'$ is an induced subgraph of $G$, it is clear that if $G$ has an $\ell$-page stack (queue) layout then $G'$ also has an $\ell$-page stack (queue) layout. To complete the proof we show that the reverse is true: If $G'$ has an $\ell$-page stack (queue) layout, then $G$ also has an $\ell$-page stack (queue) layout.

Recall that $f(\ell,p,x):=2^{2^{\ell\cdot x^2\cdot 2^{12p^2}}}$. Let $\linL$ be a fixed (but hypothetical) $\ell$-page stack (queue) layout of $G'$. %
Throughout this section, for any subgraph $H\subseteq G'$, we denote by $\linL[H]$ the specific layout of $H$ obtained by restricting $\linL$ to $H$.
If $V(G')=V(G)$ we are done. So let $V(G') = S \uplus S' \uplus L_1 \uplus \cdots \uplus L_{f(\ell, p, |S\cup S'|)}$ be a partition witnessing that $G'$ is a reduced graph. Here each $L_i$ is a $f(\ell, p, |S\cup S'|)$-large group and $S'$ is the set of all vertices in graphs that do not belong to an $f(\ell, p, |S\cup S'|)$-large equivalence class in $\sim$. For brevity, we will hereinafter use \emph{large} as shorthand for $f(\ell, p, |S\cup S'|)$-large. 

Let $\mathcal{L}=\{L_1,\cdots,L_{f(\ell,p,|S\cup S'|)}\}$. Our key idea is to identify three large groups $X,Y,Z\in \mathcal{L}$ with a special structure (a ``guiding sublayout'') in the solution $\linL$ using Ramsey theory. We will then use this pattern to insert the remaining parts of the graph. Before proceeding, we define some notations to be able to identify these groups, starting with a ``template'' $R$. 

\begin{definition}
Let $k$ be the number of large equivalence classes in $\sim$ and let $R_i\in \mathcal{C}$
be a canonical representative of the $i^{th}$ large  equivalence class $[R_i]$. Let $R:=V(R_1)\uplus \cdots \uplus V(R_k)$.
\end{definition}

Recall that by Definition~\ref{defn:large-eq-class}, if $X$ is a large group then $G[X]$ is a disjoint union of exactly one graph from each large equivalence class of $\sim$. %
We now define a natural isomorphism from $G[X]$ to $G[R]$. This will allow us to map consistently between different large groups via $R$.
\begin{definition}
For a large group $X$ with $G[X]=H_1\uplus \cdots \uplus H_k$, $H_i\in [R_i]$ for each $i\in [k]$, let $\alpha_X$ be an isomorphism from $G[X]$ to $G[R]=R_1\uplus \cdots \uplus R_k$ such that for each $i\in [k]$ and vertex $u\in H_i$, $\alpha_X(u)\in V(R_i)$, and for each $v\in S$, $uv\in E(G)$ if and only if $\alpha_X(u)v\in E(G)$.   %
\end{definition}

For any large group $A$ and for each $u\in R$ we will sometimes use $u_A$ as shorthand for $\alpha_A^{-1}(u)$. \oldtodo{R: I think this should be made consistent: right now we use both options from the previous sentence. I'd be fine with just keeping, e.g., the $-1$ notation.\\ TD: 2025-08-19: decided to keep for now both versions}
From now we fix two distinct large groups $L$ and $L'$ in $\mathcal{L}$. Further, for distinct large groups $X,Y\in \mathcal{L}$, we say $X\prec Y$ if in $\prec$, the first vertex in $X\cup Y$ is from $X$.
\begin{definition}
For $X\prec Y\in \mathcal{L}\setminus \{L,L'\}$, we define $\phi_{X,Y}:S\cup S'\cup X\cup Y \rightarrow S\cup S'\cup L\cup L'$
\begin{itemize}
	\item $\phi_{X,Y}(s)=s$ for each $s\in S\cup S'$
	\item $\phi_{X,Y}(x)=\alpha^{-1}_{L}(\alpha_X(x))$ for each $x\in X$
	\item $\phi_{X,Y}(y)=\alpha^{-1}_{L'}(\alpha_Y(y))$ for each $y\in Y$
\end{itemize}
Note that $\phi_{X,Y}$ is an isomorphism from $G[S\cup S'\cup X\cup Y]$ to $G[S\cup S'\cup L\cup L']$.
\end{definition}
\begin{definition}
For $X\prec Y\in \mathcal{L}\setminus\{L,L'\}$, 
$\textsf{info}_{\linL}(X, Y)$ is the stack (queue) layout of $G[S\cup S'\cup L\cup L']$ obtained from $\linL[G[S\cup S'\cup X\cup Y]]$ %
using the isomorphism $\phi_{X,Y}$.%
\end{definition}

We are now ready to show the existence of three distinct large groups $X,Y,Z$ with a useful consistent pattern between them and $S\cup S'$ based on $\textsf{info}$. Essentially, these three groups will form the aforementioned guiding sublayout used to argue the correctness of the pruning step, i.e., establish Lemma~\ref{lem:VI-kernel-proof} by showing that we can reinsert all the removed vertices by building on $\linL[G[S\cup S' \cup X \cup Y \cup Z]]$. Note that---perhaps counterintuitively---we will do so by discarding all other information about $\linL$. We refer to Figure~\ref{fig:ramsey-overview} for an overview of the entire approach and to \Cref{fig:ramsey-blocks} for a visualization of $\textsf{info}$.

\begin{figure}
\begin{center}
	\includegraphics[page=1]{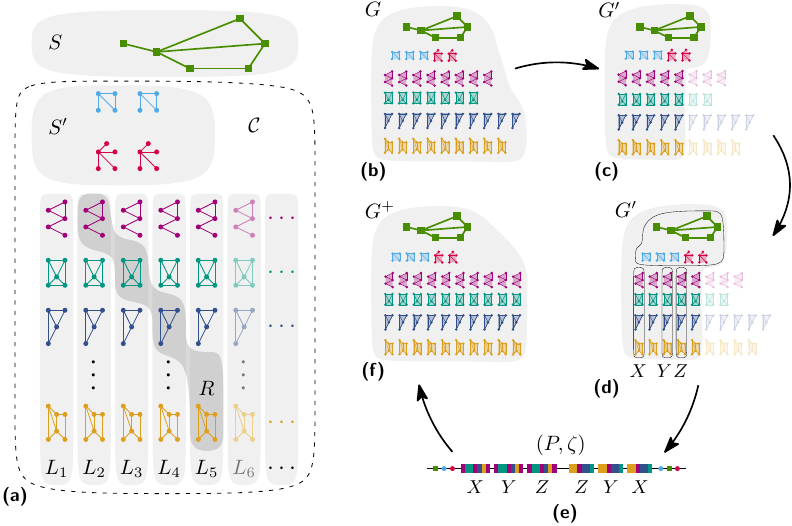}
\end{center}
\caption{A schematized overview of~\textbf{\textsf{(a)}} the used notation and \textbf{\textsf{(b)}--\textsf{(f)}} proof of \Cref{thm:vifpt}. In the latter, $G$ is the original input graph, $G'$ the reduced graph obtained after applying Lemma~\ref{lem:reducedcomputation}, and $S\cup S' \cup X \cup Y \cup Z$ form the guiding sublayout that is used to construct an $\ell$-page layout for a graph $G^+$ which is a supergraph of $G$---hence establishing that $G'$ admits an $\ell$-page stack (or queue) layout if and only if so does $G$.}
\label{fig:ramsey-overview}
\end{figure}

\begin{lemma}
\label{lem:threelarge}
There exists distinct large groups $X,Y,Z \in \mathcal{L} \setminus \{L,L'\} $ with $X\prec Y\prec Z$ such that $\textsf{info}_{\linL}(X, Y)=\textsf{info}_{\linL}(Y, Z)=\textsf{info}_{\linL}(X, Z)$. %
\end{lemma}
\begin{proof}

Construct a complete edge colored auxiliary graph $\mathcal{H}$ %
with $V(\mathcal{H})=\mathcal{L} \setminus \{L,L'\} $. For $X,Y\in \mathcal{L} \setminus \{L,L'\} $ with $X\prec Y$, we set $color((X,Y))=\mathsf{info}_{\linL}(X,Y)$. Let $x:=|S\cup S'|$.

\begin{restatable}\restateref{claim:edge-colors}{claim}{claimRamseyEdgeColors}
	\label{claim:edge-colors}
	\oldtodo{check and fix values\\ TD: still relevant?}
	The number of possible distinct edge colors of $\mathcal{H}$ is at most $2^{\ell\cdot x^2\cdot 2^{9p^2}}$.
\end{restatable}
\begin{statelater}{pclaimRamseyEdgeColors}
	\begin{claimproof}
		Recall that \linL is an $\ell$-page stack (queue) layout of $G'$. By definition, $\mathsf{info}_{\linL}(X,Y)$ for $X\prec Y\in \mathcal{L}\setminus \{L,L'\}$ is a stack (queue) layout $\langle \prec',\sigma'\rangle$ of $G[S\cup S'\cup L\cup L']$ using at most $\ell$-pages. 
		The number of possible $\ell$ page layouts of $G[S\cup S'\cup L\cup L']$ is at most $(|S\cup S'\cup L\cup L'|)!\cdot \ell^{|S\cup S'\cup L\cup L'|^2}$. There are 
		at most $p\cdot 2^{2p^2}$ equivalence classes and each graph in any equivalence class has size at most $p$. Thus we can bound the size of large groups $|L|=|L'|\leq p^2\cdot 2^{2p^2}$. Then using $|S\cup S'|=x$
		and $|L|+|L'|\leq 2p^2\cdot 2^{2p^2}$, we bound the number of possible $\ell$ page layouts of $G[S\cup S'\cup L\cup L']$ by $2^{\ell\cdot x^2\cdot 2^{9p^2}}$. 
	\end{claimproof} 
\end{statelater}

We now use a well-known fact (from Ramsey theory~\cite{Alon_asymptoticallytight}) that any edge-colored clique on $n$ vertices colored with $t$ colors has a monochromatic clique of size at least $\log_t(n)/t$. 

We have $|V(\mathcal{H})|=|\mathcal{L}|-2= f(\ell,p,x)-2\geq 2^{2^{\ell\cdot x^2\cdot 2^{11p^2}}}$. Thus $n\geq 2^{2^{\ell\cdot x^2\cdot 2^{11p^2}}}$ and $t\leq 2^{\ell\cdot x^2\cdot 2^{9p^2}}$ and therefore there is a monochromatic clique of size at least $\log_t(n)/t\geq 2^{2^2}\geq 3$.
A monochromatic clique of size $3$ in $\mathcal{H}$ corresponds to distinct large groups $X,Y,Z\in \mathcal{L}\setminus \{L,L'\}$ with $X\prec Y\prec Z$ such that $\textsf{info}_{\linL}(X, Y)=\textsf{info}_{\linL}(Y, Z)=\textsf{info}_{\linL}(X, Z)$. 
\end{proof}
Let $X,Y,Z\in \mathcal{L}$ be three large groups witnessing the previous lemma. We now show that we can partition the vertices of $X$, $Y$, and $Z$ in a nice way that will guide us on how to insert the remaining parts of the graph. 
\onlyShort{
Towards this, we first note that the $\textsf{info}_{\linL}$ forces (1) edges incident to $X$, $Y$ and $Z$ to have the same page assignment, and moreover (2) forces the individual vertex orderings of $X\cup S\cup S'$, $Y\cup S\cup S'$, $Z\cup S\cup S'$ to be exactly the same as 
some ordering $\prec'$ of $R\cup S\cup S'$; see also  \ifthenelse{\boolean{cameraready}}{the full version~\cite[Observations~18 and 19]{ARXIV}}{Observations~\ref{obs:pagecolor} and~\ref{obs:internalorder} in the appendix}. However, this does not yet fix how the vertices of $X$, $Y$ and $Z$ are ordered with respect to each other.
}

\begin{statelater}{viBlockObservations}
We formally observe that the isomorphic copies of the edges incident to $X$, $Y$ and $Z$ have the same page assignment in $\linL$ because $\textsf{info}_{\linL}(X, Y)=\textsf{info}_{\linL}(Y, Z)=\textsf{info}_{\linL}(X, Z)$. Recall for $u\in R$ and $A\in \mathcal{L}$, $u_A=\alpha^{-1}_A(u)$.
\begin{observation}
	\label{obs:pagecolor}
	For each $u\in R$:        
	\begin{itemize}
		\item For each $v\in R$ such that $(u,v)\in E(G)$ we have $\sigma(u_X,v_X)=\sigma(u_Y,v_Y)=\sigma(u_Z,v_Z)$;
		\item For each $s\in S$ such that $(u,s)\in E(G)$ we have $\sigma(u_X,s)=\sigma(u_Y,s)=\sigma(u_Z,s)$.
	\end{itemize} 
\end{observation}

Similarly, $\prec_{X\cup S\cup S'}$, $\prec_{Y\cup S\cup S'}$ and $\prec_{Z\cup S\cup S'}$ must behave in the same way when viewed as isomorphic copies with respect to $R$ due to definition of $\textsf{info}_{\linL}$; we now formalize this.

\begin{observation}
	\label{obs:internalorder}
	There is an ordering $\prec'$ of vertices of $R\cup S\cup S'$ such that:
	\begin{itemize}
		\item For each $s\in S\cup S'$ and $r\in R$, if $s\prec' r$ then $s\prec r_X$, $s\prec r_Y$, and $s\prec r_Z$;
		\item For each $s\in S\cup S'$ and $r\in R$, if $r\prec' s$ then $ r_X\prec s$, $r_Y\prec s$, and $r_Z\prec s$;
		\item For each $u,v\in S\cup S'$, if $u\prec' v$ then $u\prec v$;
		\item For each $u,v\in R$ if $u\prec'v$ then $u_X\prec v_X$, $u_Y\prec v_Y$ and $u_Z\prec v_Z$.
	\end{itemize} %
\end{observation}

Note that Observation~\ref{obs:internalorder} does not provide any information about how the vertices of $X$, $Y$ and $Z$ are ordered with respect to each other. Obtaining some structure in that regard will be our next, major task.
\end{statelater}

For the following, let us set $\Upsilon=S\cup S' \cup X \cup Y \cup Z$. 
Let a \emph{block} be a consecutive subsequence of vertices from $R$ w.r.t.\ $\prec'$, and a \emph{solution-block} of $\linL[G[\Upsilon]]$ be a consecutive subsequence of vertices from $\Upsilon$; see also \Cref{fig:ramsey-blocks-example}.
\begin{figure}[ht]
\begin{center}
	\includegraphics[page=1]{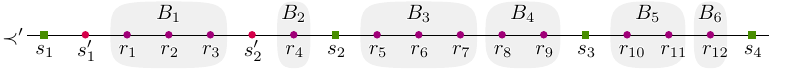}
\end{center}
\caption{An ordering $\prec'$ of $R \cup S \cup S'$ and a partition of $R$ into the blocks $B_1$ to $B_6$, indicated in gray.
	Note that vertices of $S \cup S'$, colored green and red, respectively, always separate two blocks. However, not every pair of blocks is separated by such a vertex, see, for example, blocks $B_3$ and $B_4$.}
\label{fig:ramsey-blocks-example}
\end{figure}

\newcommand{\asc}{\ensuremath{\nearrow}}
\newcommand{\desc}{\ensuremath{\searrow}}

We are now ready to prove the structural result which provides strong guardrails on $\linL[G[\Upsilon]]$; in particular, it guarantees that all vertices from $X$, $Y$ and $Z$ must occur in consecutive solution-blocks and following a strict ``ascending'' or ``descending'' order. We remark that this structural result is only made possible by the extraction of three copies of large groups via Lemma~\ref{lem:threelarge}; see also 
\Cref{fig:ramsey-blocks} for an illustration. 

\begin{restatable}\restateref{lem:structure}{lemma}{lemmaStructure}
\label{lem:structure}
There exists a pair $(P,\zeta)$, where $P$ is a partitioning of $R$ into blocks and $\zeta\colon P\rightarrow\{\asc, \desc\}$ satisfying the following.
For every block $B\in P$ with $\zeta(B) = \asc$, the sequence $\alpha^{-1}_X(B)\alpha^{-1}_Y(B)\alpha^{-1}_Z(B)$ is a solution-block of $\linL[G[\Upsilon]]$.
Moreover, for every block $B\in P$ with $\zeta(B) = \desc$, the sequence $\alpha^{-1}_Z(B)\alpha^{-1}_Y(B)\alpha^{-1}_X(B)$ is a solution-block of $\linL[G[\Upsilon]]$.
\end{restatable}
\begin{prooflater}{plemmaStructure}
By definition of $X,Y,Z$, we know $\textsf{info}_{\linL}(X,Y)=\textsf{info}_{\linL}(Y,Z)=\textsf{info}_{\linL}(X,Z)$. Using this relationship and the definition of $\textsf{info}$,  we have the following direct properties, where we recall that for a large group $A$ and $u \in R$, we use $u_A$ as a shorthand for $\alpha_A^{-1}(u)$:

\begin{itemize}
	\item For each $u\in R$, either $u_X\prec u_Y\prec u_Z$ or $u_Z\prec u_Y\prec u_X$.
	\item For any block $B$ of vertices in R, if all vertices in $\alpha^{-1}_{X}(B)$ occur before all vertices in $\alpha^{-1}_{Z}(B)$ then all vertices in $\alpha^{-1}_{Y}(B)$ also occur before all vertices in $\alpha^{-1}_{Z}(B)$.
	\item For any block $B$ of vertices in R, if all vertices in $\alpha^{-1}_{Z}(B)$ occur before all vertices in $\alpha^{-1}_{X}(B)$ then all vertices in $\alpha^{-1}_{Y}(B)$ also occur before all vertices in $\alpha^{-1}_{X}(B)$.
\end{itemize}

On a high level we traverse the linear order $\prec$ of $G'$ from left to right. We find the first vertex $u \in X\cup Y \cup Z$. If $u$ belongs to $X$, we then traverse $\prec$ and build a solution block $\alpha^{-1}_X(B)$ containing only $X$ vertices. This corresponds to block $B$ of $R$. We show below that $\alpha^{-1}_X(B)\alpha^{-1}_Y(B)\alpha^{-1}_Z(B)$ is a solution block. We then add $B$ to $P$ and set $\zeta(B)=\asc$ and continue the traversal. If instead $u$ was a vertex in $Z$ then we would have found a descending block. We now formalize this argument.

Let $R':=R$ and let $u$ be the first vertex in $R'$ in $\prec'$. Further let $\alpha^{-1}_X(u)\prec \alpha^{-1}_Y(u) \prec \alpha^{-1}_Z(u)$. 
Choose $B=u\prec'\cdots \prec'w$ to be the largest block of $R$ starting with $u$ such that $\alpha^{-1}_X(B)$ is a solution block, and let $p$ be the direct successor of $w$ according to $\prec'$ in $R$. That is we traverse $\prec$ starting at $\alpha^{-1}_X(u)$ until we find a vertex not in $X$. Note $u$ could be the same as $w$. We now show that $\alpha^{-1}_X(B)\alpha^{-1}_Y(B)\alpha^{-1}_Z(B)$ forms a solution block.

First we can infer that all vertices in $\alpha^{-1}_{X}(B)$ occur before all vertices in $\alpha^{-1}_{Z}(B)$. This implies all vertices in $\alpha^{-1}_{Y}(B)$ also occur before all vertices in $\alpha^{-1}_{Z}(B)$. Therefore $\alpha^{-1}_X(B)\alpha^{-1}_Y(u)$ is also a solution block. Next we show that $w_Z\prec p_X$.  

We know by the properties we observed earlier that either $p_Z\prec p_Y\prec p_X$ or $p_X\prec p_Y\prec p_Z$. If it is the former then $w_Z\prec p_Z\prec p_X$ as desired. If it is the latter, then $p_X$ is somewhere between $u_Y$ and $w_Z$ with $u_Y\prec w_Y\prec u_Z \prec w_Z$.
We can divide this into two cases -- $u_Y\prec p_X\prec u_Z$ or $w_Y\prec p_X\prec w_Z$ -- but this would in either case imply that $\textsf{info}(X,Y)=\textsf{info}(Y,Z)=\textsf{info}(X,Z)$ is not true. This shows that $\alpha^{-1}_X(B)\alpha^{-1}_Y(B)\alpha^{-1}_Z(B)$ is a solution block. We add $B$ to $P$ and set $\zeta(B)=\asc$. 

If we were in the case where $\alpha^{-1}_Z(u)\prec \alpha^{-1}_Y(u) \prec \alpha^{-1}_X(u)$, similar arguments will hold with $\alpha^{-1}_Z(B)\alpha^{-1}_Y(B)\alpha^{-1}_X(B)$ being a solution block and we set $\zeta(B)=\desc$. 
Finally recursing with $R'=R'\setminus B$ completes the proof.
\end{prooflater}

Crucially, we can now use Lemma~\ref{lem:structure} to safely insert an arbitrary number of large groups into $G[\Upsilon]$ without increasing the number of required pages. 

\begin{figure}
\begin{center}
	\includegraphics[page=1]{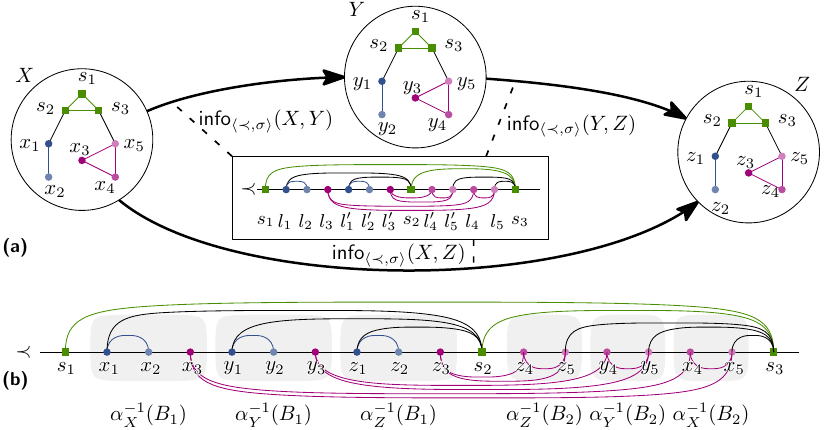}
\end{center}
\caption{\textbf{\textsf{(a)}} The three large groups $X \prec Y \prec Z \in \mathcal{L} \setminus \{L, L'\}$ form a monochromatic clique in the auxiliary graph $\mathcal H$, i.e., we have $\textsf{info}(X,Y)=\textsf{info}(Y,Z)=\textsf{info}(X,Z)$. 
	The tuple $(P, \zeta)$, where $P$ partitions $R$ into two blocks $B_1 = \{r_1, r_2, r_3\}$ and $B_2 = \{r_4, r_5\}$ with $\zeta(B_1) = \asc$ and $\zeta(B_2) = \desc$ fulfills the properties from \Cref{lem:structure}.
	As $X, Y, Z$ form a monochromatic clique, the solution blocks from~\textbf{\textsf{(b)}} must occur in $\linL[G[\Upsilon]]$.}
\label{fig:ramsey-blocks}
\end{figure}

\begin{restatable}\restateref{lem:layout-supergraph}{lemma}{lemmaLayoutSupergraph}
\label{lem:layout-supergraph}
Let $G^+$ be the supergraph of $G$ with each large equivalence class having equal size, then $G^+$ has an $\ell$-page stack (queue) layout. 
\end{restatable}
\begin{proofsketch}
Let $t\geq 3$ be the maximum size of a large equivalence class of $\sim$. %
Note that $G^+$ can be expressed as $G^+=G^+[S\cup S'\cup X_1\cup  \cdots \cup X_t]$ where each $X_i,i\in[t]$ is a large group. 
We construct a stack (or queue, as desired) layout $\linLs$ of $G^+$ as follows.

\proofsubparagraph*{Construction of $\sigma^+$:} For each edge $(u,v)\in G^+$,
\begin{itemize}
	\item If $u,v\in S\cup S'$, we set $\sigma^+(u,v):=\sigma(u,v)$.
	\item If $u\in X_i$ for some $i\in [t]$ and $v\in S$, we set $\sigma^+(u,v):=\sigma(r_Y,v)$ where $r=\alpha_{X_i}(u)$.
	\item If $u,v\in X_i$ for some $i\in [t]$, we set $\sigma^+(u,v):=\sigma(r_Y,r'_Y)$ where $r=\alpha_{X_i}(u)$, $r'=\alpha_{X_i}(v)$.
\end{itemize}
\proofsubparagraph*{Construction of $\prec^+$:} 
In $\prec'$, replace each block $B\in P$ with $\alpha_{X_1}^{-1}(B)\alpha_{X_2}^{-1}(B)\cdots\alpha_{X_t}^{-1}(B)$ if $\zeta(B)= \asc $ and with $\alpha_{X_t}^{-1}(B)\alpha_{X_{t-1}}^{-1}(B)\cdots\alpha_{X_1}^{-1}(B)$ if $\zeta(B)= \desc $.

\medskip
To show that $\linLs$ is indeed an $\ell$-page stack (queue) layout of $G^+$, we show that any conflict in $\linLs$ would also imply a conflict in $\linL[G[\Upsilon]]$ in the \ifthenelse{\boolean{cameraready}}{full proof~\cite{ARXIV}}{\hyperref[lem:layout-supergraph*]{full proof}}.
\end{proofsketch}
\begin{prooflater}{proofLayoutSupergraph}
Let $t\geq 3$ be the maximum size of a large equivalence class of $\sim$. %
Hence, $G^+$ can be expressed as $G^+=G^+[S\cup S'\cup X_1\cup  \cdots \cup X_t]$ where each $X_i,i\in[t]$ is a large group.  

Let $\linLs$ be the stack (queue) layout of $G^+$ that we will construct. %
Recall that there are no edges across large groups or from large groups to $S'$. Further $\prec'$ is the ordering of $R\cup S\cup S'$ from Observation~\ref{obs:internalorder}. Also let $(P,\zeta)$ be a pair witnessing Lemma~\ref{lem:structure} where $P$ is a partitioning of $R$ into blocks and $\zeta\colon P\rightarrow\{\asc, \desc\}$. %
\proofsubparagraph*{Construction of $\sigma^+$:} For each edge $(u,v)\in G^+$,
\begin{itemize}
	\item If $u,v\in S\cup S'$, we set $\sigma^+(u,v):=\sigma(u,v)$.
	\item If $u\in X_i$ for some $i\in [t]$ and $v\in S$, we set $\sigma^+(u,v):=\sigma(r_Y,v)$ where $r=\alpha_{X_i}(u)$.
	\item If $u,v\in X_i$ for some $i\in [t]$, we set $\sigma^+(u,v):=\sigma(r_Y,r'_Y)$ where $r=\alpha_{X_i}(u)$, $r'=\alpha_{X_i}(v)$.
\end{itemize}
\proofsubparagraph*{Construction of $\prec^+$:} 
In $\prec'$, replace each block $B\in P$ with $\alpha_{X_1}^{-1}(B)\alpha_{X_2}^{-1}(B)\cdots\alpha_{X_t}^{-1}(B)$ if $\zeta(B)= \asc $ and with $\alpha_{X_t}^{-1}(B)\alpha_{X_{t-1}}^{-1}(B)\cdots\alpha_{X_1}^{-1}(B)$ if $\zeta(B)= \desc $.
\\ \\
We now show that $\linLs$ is an $\ell$-page stack (queue) layout of $G^+$. 
Let $\Upsilon=S\cup S'\cup X\cup Y\cup Z$. By using Observation~\ref{obs:pagecolor} and the definition of $\prec'$, $(P,\zeta)$ and $\linLs$, one can verify that for each $a<b<c\in [t]$, $\linLs[G^+[S\cup S'\cup X_a\cup X_b\cup X_c]]$ is an $\ell$-page stack(queue) layout.
This is because $G^+[S\cup S'\cup X_a\cup X_b\cup X_c]$ is isomorphic to $G[\Upsilon]$ and $\linL[G[\Upsilon]]$ is the stack(queue) layout obtained from $\linLs[G^+[S\cup S'\cup X_a\cup X_b\cup X_c]]$ by mapping $X_a,X_b,X_c$ to $X,Y,Z$ via $R$.

Now suppose there are two edges $(a,b)$ and $(c,d)$ that conflict with each other in $\linLs$. Since there are no edges across large groups. $a,b,c,d\in X_i\cup X_j\cup S\cup S'$ for some $i<j\in [t]$. But this implies $\linLs[G^+[S\cup S'\cup X_i\cup X_j]]$ is not an $\ell$-page stack(queue) layout of $G^+[S\cup S'\cup X_i\cup X_j]$, which is a contradiction.
\end{prooflater}

\section{An $\textsf{XP}$-Algorithm For Bounded-Width Linear Layouts}\label{sec:poly-pw}
In this section, we present an algorithm to decide if a given graph $G$ has a stack (or queue) layout on $\ell$ pages and with page width $q$, i.e., an $\ell$-page $q$-width linear layout, in $n^{\BigO{q\cdot\ell}}$ time.
In the affirmative case, it can also output such a layout.
Since only little adaption is required between stack and queue layouts, we only discuss problem specific changes where necessary.
We will first give an intuitive overview over our approach and the core ideas to showing its correctness, details can be found in \shortLong{\Cref{sec:poly-state-graph-short,sec:decompose-cut-set} and the full version}{\Cref{sec:poly-state-graph,sec:decompose-cut-set}}.

\subparagraph*{Bounded-size Interfaces.}
Observe that, for each vertex $v \in V$ in a stack (or queue) layout \linL of $G$, the edges $\EdgesCrossingV[\prec]{v}$ that span the spine right next to $v$, induce the cut $(\LeftAtOfV[\prec]{v}, \RightOfV[\prec]{v})$ in $G$; see \Cref{fig:poly-pw-nicely-oriented-cut-sets}a.
The size of the cut-set $F = \EdgesCrossingV[\prec]{v}$ only depends on the number of pages and their width.
Our algorithm is centered around the insight that, through limiting these parameters and thus the size of any such cut, we limit the size of the ``interface'' (a notion we will define more precisely in a second) between a drawing on its left side and a drawing on its right side (\Cref{lem:poly-pw-cuts}).
We note that this insight parallels to the approach used by Saxe~\cite{Sax.DPA.1980} to test if a graph has bandwidth at most $k$.
In our setting, replacing the left side of such a cut in any solution drawing with one that provides the same interface will yield another solution.
This is because having the same interface will imply inducing the same decompositions into left and right vertices (\Cref{lem:poly-pw-cut-inudce-unique}), the statement then follows by simply gluing the two drawings together along the cut.

Our bounded-size interface consists of the set of edges $F$ that intersect the cut right next to $v$, their (vertical) order along the cut, and the information which edge endpoint lies on which side of the cut, together with a page assignment for the edges in $F$.
We will encapsulate this information in what we call a (nicely) \emph{oriented cut-set}; see \Cref{def:nicely-oriented-cut}.

\begin{figure}
	\centering
	\includegraphics[page=1]{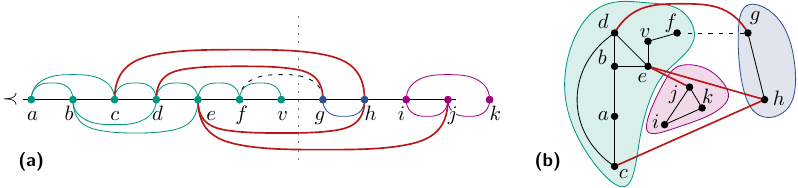}
	\caption{
		The two-page width-five stack layout in~\textbf{\textsf{(a)}} induces the cut-set $F = \EdgesCrossingV[\prec]{v}$ indicated with the red edges.
		\textbf{\textsf{(b)}} The three connected components in $G \setminus F$ are indicated with color.
		Observe that $(F, c \prec d \prec e \prec g \prec h \prec j)$ is nicely oriented.
		The source $e$ and sink $g$ cannot be in the same connected component, as every path between them contains at least one edge of $F$; e.g., the edge $fg$ cannot exist\onlyLong{; see also the proof of \Cref{lem:poly-pw-cuts}}.
}
	\label{fig:poly-pw-nicely-oriented-cut-sets}
\end{figure}

\subparagraph*{Our Dynamic Programming Solution.}
Given only an arbitrary oriented cut-set $F$ for an instance $G$, we are able to infer which connected components and thus which vertices of the graph lie on the left (or also right) side of the cut (\Cref{lem:poly-pw-cut-side}; see also \Cref{fig:poly-pw-nicely-oriented-cut-sets}b).
This now allows us to solve this problem via dynamic programming.
We build a directed acyclic state graph where each state corresponds to one configuration of our interface, i.e. one (nicely) oriented cut-set together with a page assignment for its edges (formalized as \ifthenelse{\boolean{cameraready}}{$\mathcal{N}$C 1 and 2 in the full version~\cite{ARXIV}}{\Cref{nc:nicely-oriented} and \ref{nc:pagewidth} in \Cref{sec:poly-state-graph}}).
We add an arc between two states if we can move from one oriented cut-set to the other by moving exactly one vertex (which we inferred to be on the right side via \Cref{lem:poly-pw-cut-side}) to the (end of a linear order used on the) left without directly violating a necessary stack or queue layout property (i.e., no crossing or nesting of edges) with regard to only the edges in the current cut-set; the exact conditions are again formalized as \ifthenelse{\boolean{cameraready}}{$\mathcal{A}$C 1--4 in the full version~\cite{ARXIV}}{\Cref{ac:directed-cut-set}--\ref{ac:transition} in \Cref{sec:poly-state-graph}}.
In \Cref{lem:poly-pw-feasible-states}, we will show that this means that we can extend any partial solution represented by the current state with the further vertex when moving along such an arc to obtain a solution for the new state.
Then, there exists a solution to our instance if and only if there exists a directed path between the (artificially-distinguished) empty beginning and end states.
As finding such a path takes linear time in the size of the state graph, \Cref{thm:pagewidth} then follows from the insight that the state graph has a size bounded by $n^{\BigO{q \cdot \ell}}$ and can be computed in this time.

\onlyLong{In the following sections, we formalize the above-provided intuitive overview.}

\subsection{Decomposing Linear Layouts Using Oriented Cut-Sets}\label{sec:decompose-cut-set}
	As already hinted above, our algorithm uses at its core cut-sets of $G$ to compute the desired stack (or queue) layout \linL[G].
	Before we discuss their usage in detail, 
	let us first take a closer look at them and define their desired properties.

	\begin{observation}\label{obs:cut-set}
    Given a connected graph $G$ and a cut-set $F \subseteq E$ that separates the graph into $k \geq 2$ connected components $C_1,\ldots,C_k$.
    For all $i\in [k]$, component $C_i$ contains at least one vertex that is an endpoint of an edge in $F$.
	\end{observation}
We consider in addition to a cut-set $F$ a total order $\prec_F$ on the endpoints of its edges and call this the \emph{oriented cut-set}, denoted as $\DirectedCut{F}$.
While a single cut-set $F$ has \BigO{(2\Size{F})!} different oriented cut-sets, we are only interested in ``nicely oriented'' cut-sets.
Let $F$ be a cut-set and $\prec_F$ a linear order on the endpoints of edges in $F$.
A vertex $v \in V(F)$ is a \emph{source} (or \emph{sink}) in \DirectedCut{F} if for every edge $uv \in F$ we have $v \prec_{F} u$ ($u \prec_{F} v$, respectively).

\begin{definition}\label{def:nicely-oriented-cut}
  For a cut-set $F$, $\DirectedCut{F}$ is \emph{nicely oriented} if and only if each vertex $v \in V(F)$ is either a source or sink, each connected component of $G - F$ contains either only sources or only sinks, and all sources are left of every sink in $\prec_{F}$.
\end{definition}

To ease presentation, we %
allow the linear order $\prec_F$ of an oriented cut-set \DirectedCut{F} to be a total order of a super-set of $V(F)$. %
We consider $\DirectedCut[\prec]{F}$ and $\DirectedCut[\prec']{F}$ identical if $\prec$ and $\prec'$ agree on $V(F)$.
We now establish a crucial connection between linear layouts and nicely oriented cut-sets of $G$; see also \Cref{fig:poly-pw-nicely-oriented-cut-sets}.

\begin{restatable}\restateref{lem:poly-pw-cuts}{lemma}{lemmaPolyPwCuts}
	\label{lem:poly-pw-cuts}
	Let $G$ be a graph and \linL be an $\ell$-page $q$-width linear layout of $G$.
	For every vertex $v \in V$ that is not the rightmost in $\prec$, the edges $\EdgesCrossingV[\prec]{v}$ form a cut-set of size at most $q\cdot\ell$ that induces the cut $(\LeftAtOfV[\prec]{v}, \RightOfV[\prec]{v})$ in $G$.
	Furthermore, $\DirectedCut[\prec]{\EdgesCrossingV[\prec]{v}}$ is nicely oriented.
\end{restatable}
\begin{prooflater}{plemmaPolyPwCuts}
	Let \linL be an $\ell$-page $q$-width stack (or queue) layout of $G$.
	We consider an arbitrary vertex $v \in V$ that is not the rightmost in $\prec$ and first show that $\EdgesCrossingV{v}$ is a cut-set of size at most $q\cdot\ell$ inducing the cut $(\LeftAtOfV{v}, \RightOfV{v})$.
	Afterwards, we show that $\DirectedCut[\prec]{\EdgesCrossingV{v}}$ is nicely oriented.
	
	Towards a contradiction, assume that $\EdgesCrossingV{v}$ is not a cut-set of $G$.
	In particular, this means that there exists two vertices $u \in \LeftAtOfV{v}$ and $w \in \RightOfV{v}$ that are adjacent in $G - \EdgesCrossingV{v}$.
	As $u \in \LeftAtOfV{v}$ and $w \in \RightOfV{v}$, we have $u \preceq v$ and $v \prec w$.
	However, this implies $uw \in \EdgesCrossingV{v}$ by the definition of \EdgesCrossingV{v}, i.e., these two vertices cannot exist.
	Hence, $\LeftAtOfV[\prec]{v}$ and $\RightOfV[\prec]{v}$ are (at least) two connected components in $G - \EdgesCrossingV{v}$, which means that $\EdgesCrossingV{v}$ is a cut-set inducing the cut $(\LeftAtOfV[\prec]{v}, \RightOfV[\prec]{v})$.
	The size of $\EdgesCrossingV{v} \leq q \cdot \ell$ follows directly from the definition of page width.

	Consider the oriented cut-set $\DirectedCut[\prec]{\EdgesCrossingV{v}}$.
	Towards a contradiction, assume that $\DirectedCut[\prec]{\EdgesCrossingV{v}}$ is not nicely oriented.
	This means that either a vertex is both a source and a sink, there exists a connected component of $G - \EdgesCrossingV{v}$ that contains a source and a sink, or some source vertex $u$ is right of a sink vertex $v$ in $\prec$.
	The first and last case cannot occur by the construction of $\DirectedCut[\prec]{\EdgesCrossingV{v}}$: The vertices left of (and including) $v$ are sources, all others are sinks.
	Thus, assume that there exists a connected component $C$ with a source $u_1$ and a sink $u_2$.
	Hence, we can identify two edges $u_1w_1, w_2u_2 \in \EdgesCrossingV{v}$ such that $u_1, u_2 \in V(C)$, $w_1, w_2 \notin V(C)$, $u_1 \prec w_1$, and $w_2 \prec u_2$ holds.	
	It holds  $u_1 \preceq v \prec u_2$ by the definition of$\DirectedCut[\prec]{\EdgesCrossingV{v}}$.
	As $C$ is a connected component, there is a path $P$ in $C$ connecting $u_1$ and $u_2$.
	However, then we can find at least one edge $u_3w_3$ on $P$ such that $u_3 \preceq v \prec w_3$; see for example $fg$ in \Cref{fig:poly-pw-nicely-oriented-cut-sets}.
	By definition, we have then $u_3w_3 \in \EdgesCrossingV{v}$, which contradicts the existence of $P$ in $C$.
	As this holds for all such paths, we conclude that $\DirectedCut[\prec]{\EdgesCrossingV{v}}$ is nicely oriented.
\end{prooflater}

A liner order $\prec$ \emph{induces} an oriented cut-set $\DirectedCut[\prec']{F}$ if there exists a vertex $v$ that witnesses $\DirectedCut[\prec]{\EdgesCrossingV[\prec]{v}} = \DirectedCut[\prec']{F}$.
Note that %
the same oriented cut-set can be induced by different linear orders.
From \Cref{lem:poly-pw-cuts}, we can deduce that all cuts $(A, B)$ of an $\ell$-page $q$-width stack (and queue) layout of $G$ can be constructed using at most $\binom{m}{\ell \cdot \pw} = \BigO{m^{\ell \cdot \pw}}$ different cut-sets.
We aim to use their oriented counterparts to obtain the spine order $\prec$ of the desired stack (or queue) layout.
We next show that we can efficiently compute for a given oriented cut-set $\DirectedCut{F}$ a witnessing linear order that induces $\DirectedCut{F}$.
\begin{restatable}\restateref{lem:poly-pw-cut-side}{lemma}{lemmaPolyPwCutSide}				\label{lem:poly-pw-cut-side}
	Given a graph $G$ and a nicely oriented cut-set \DirectedCut{F} of $G$, we can compute a linear order $\prec$ of $V$ that induces \DirectedCut{F} in \BigO{n + m} time.
\end{restatable}
\begin{prooflater}{plemmaPolyPwCutSide}
  Let \DirectedCut{F} be a nicely oriented cut-set of $G$.
  Consider the $k \geq 2$ connected components $C_1, \ldots, C_k$ of $G - F$.
  Every $C_i$, $i \in [k]$, contains at least one endpoint of an edge in $F$ and as \DirectedCut{F} is nicely oriented, every $C_i$ contains either only sources or only sinks.
  We can assume that the connected components are ordered such that, for some $k' \leq k$, all $C_i$, $1 \leq i \leq k'$, contain sources and all $C_j$, $k' < j \leq k$, contain sinks.

  For every $i \in [k]$, let $\prec_i$ be an arbitrary linear order of $V(C_i) \setminus V(F)$.
  To construct $\prec$, we take the transitive closure of $\prec_1 \oplus \ldots \oplus \prec_{k'} \oplus \prec_F \oplus \prec_{k' + 1} \oplus  \ldots \oplus \prec_k$, where $\oplus$ indicates the concatenation of total orders on pairwise disjoint sets.
  We can obtain the connected components of $G - F$ and the orders $\prec_i$ in \BigO{m} time by using, for example, a breadth-first search (BFS).
  The construction of $\prec$ takes $\BigO{\sum_{i \in [k]}\Size{V(C_i)}} + \BigO{\Size{F}} = \BigO{n}$ time.
  
  We now show that $\DirectedCut[\prec]{\EdgesCrossingV[\prec]{v}} = \DirectedCut{F}$ holds, where $v$ is the rightmost source in \DirectedCut{F}.
  As $\prec_F$ occurs as a suborder in $\prec$, we only need to show equality of the two sets.
  By the same argument, $F \subseteq \EdgesCrossingV[\prec]{v}$ follows immediately.
  We thus focus on the other direction and consider an arbitrary edge $uw \in \EdgesCrossingV[\prec]{v}$ for which we assume, for the sake of a contradiction, that $uw \notin F$ holds.
  As $F$ is a cut-set, $u, w \in V(C_i)$ must hold for some $i \in [k]$, i.e., the endpoints are in the same connected component $C_i$.
  However, then, by the construction of $\prec$ and as $v$ is the rightmost source in $\prec_F$, we have either $u,w \preceq v$ (if $i \leq k'$), or $v \prec u,w$ (if $k' < i$).
  This contradicts $uw \in \EdgesCrossingV[\prec]{v}$, which requires $u \preceq v \prec w$ or $w \preceq v \prec u$.
 	Hence, $\EdgesCrossingV{v} = F$ holds and combined with the arguments from above, this implies that $\prec$ induces \DirectedCut{F}.    
\end{prooflater}

\Cref{lem:poly-pw-cut-side} only shows that we can compute \emph{some} linear order $\prec$ of $V$ that induces $\DirectedCut{F}$.
However, to efficiently use oriented cut-sets in our state graph, we must be able to precisely determine which vertices have already been placed on the spine without storing them explicitly.
The following lemma lays the foundation for this, as it shows that no matter which linear order $\prec$ we choose to induce an oriented cut-set \DirectedCut{F}, we obtain the same partition into left and right vertices.

\begin{restatable}\restateref{lem:poly-pw-cut-inudce-unique}{lemma}{lemmaPolyPwCutInduce}
	\label{lem:poly-pw-cut-inudce-unique}
	Let \DirectedCut{F} be a nicely oriented cut-set. Furthermore, let $\prec$ and $\prec'$ be two linear orders such that $\DirectedCut[\prec]{\EdgesCrossingV[\prec]{v}} = \DirectedCut[\prec']{\EdgesCrossingV[\prec']{u}} = \DirectedCut{F}$ for some $u, v \in V$, i.e., they induce $\DirectedCut{F}$.
	Then we have $\LeftAtOfV[\prec]{v} = \LeftAtOfV[\prec']{u}$.
\end{restatable}
\begin{prooflater}{plemmaPolyPwCutInduce}	
  We first show that $(\LeftAtOfV[\prec]{v}) \subseteq (\LeftAtOfV[\prec']{u})$ holds and let $w \in \LeftAtOfV[\prec]{v}$ be an arbitrary vertex, i.e., we have $w \preceq v$.
  For the sake of a contradiction, let us assume $w \notin (\LeftAtOfV[\prec']{u})$, i.e., $w \in \RightOfV[\prec']{u}$ and thus $u \prec' w$ holds.
  First, observe that the vertex $w$ cannot be an endpoint of an edge in $F$, as it would then be a source in $\DirectedCut[\prec]{\EdgesCrossingV[\prec]{v}}$ but a sink in $\DirectedCut[\prec']{\EdgesCrossingV[\prec']{u}}$, which is impossible.
  Consider now the graph $G - F$ and the connected component $C$ with $w \in V(C)$.
  $C$ contains at least one vertex, let it be $x$, that is an endpoint of an edge in $F$.
  As both $w$ and $x$ are in the same connected component $C$, there exists a path $P$ in $C$ between $w$ and $x$.
  Since $w \preceq v$ holds, we conclude that $y \preceq v$ holds for every vertex $y$ on the path $P$.
  Otherwise, at least one edge of $P$ would be in $\EdgesCrossingV[\prec]{v}$ and thus contradict the existence of $P$ in $C$.
  Thus, in particular $x \preceq v$ must hold.
   
  Consider now $\prec'$ and $u$ and recall that we have $\DirectedCut[\prec']{\EdgesCrossingV[\prec']{u}} = \DirectedCut{F}$.
 	Furthermore, observe that the existence of the path $P$ only depends on $F$ and thus we can traverse its vertices also in $\prec'$.
 	Recall that $x \preceq v$ holds, which implies that $x$ is a source in $C$, and hence we have $x \preceq' u$.
 	Furthermore, we have $w \notin (\LeftAtOfV[\prec']{u})$ by assumption, which means $x \preceq' u \prec' w$.
 	This implies that there exists an edge $zy$ in $P$ such that $y \preceq' u \prec' z$ holds.
 	However, then $yz \in \DirectedCut[\prec']{\EdgesCrossingV[\prec']{u}}$ holds, i.e., is an edge of the cut-set $F$ and, therefore, the path $P$ cannot exist in $G - F$.
 	Since the path $P$ was selected arbitrarily, above arguments hold for any path in $C$ between $w$ and $x$ and they are not in the same connected component in $G - F$.
 	This contradicts the fact that $\prec$ and $\prec'$ induce $\DirectedCut{F}$.
 	We conclude that $w \preceq' u$ must hold.
 	Hence, $(\LeftAtOfV[\prec]{v}) \subseteq (\LeftAtOfV[\prec']{u})$.
 	The direction $(\LeftAtOfV[\prec']{u}) \subseteq (\LeftAtOfV[\prec]{v})$ is symmetric. 	
\end{prooflater}

\begin{statelater}{stategraphFullDetails}
\subsection{Using Oriented Cut-Sets to Construct a State Graph}\label{sec:poly-state-graph}
We now combine \Cref{lem:poly-pw-cuts,lem:poly-pw-cut-side} to construct our state graph $H$ with vertex set $N$, from now on called nodes to distinguish them from the vertices of $G$, and arc set $A$.
Every node represents a nicely ordered cut-set \DirectedCut{F} of $G$ with $F \subseteq E$ and $1 \leq \Size{F} \leq q \cdot \ell$ together with an assignment $\sigma_{F}$ of the edges $F$ to the $\ell$ pages.
We only consider \emph{consistent} nodes, i.e., where it holds that
\begin{enumerate}[($\mathcal{N}$C1)]
	\item \DirectedCut{F} is nicely oriented, and\label[NC]{nc:nicely-oriented}
	\item for every $p \in [\ell]$ we have %
    $\Size{\sigma_F^{-1}(p)} \leq q$
    \label[NC]{nc:pagewidth}
\end{enumerate}

For a node $S = (F, \prec_F, \sigma_F) \in N$, we let $(\StateProcessed{S}, \StateNotProcessed{S})$ denote the cut induced by the (directed) cut-set \DirectedCut{F}, where we fix \StateProcessed{S} to be those vertices that have already been processed and \StateNotProcessed{S} to those that have not yet been processed.
Note that these sets correspond to $\LeftAtOfV[\prec]{v}$ and $\RightOfV[\prec]{v}$ for a linear order $\prec$ that induces $F$, i.e., $\DirectedCut[\prec]{\EdgesCrossingV[\prec]{v}} = \DirectedCut{F}$ holds.
Thanks to \Cref{lem:poly-pw-cut-side}, we can compute these sets in linear time.
We add two dummy nodes $S_{\emptyset}$ and $S_{V}$ to $N$ that correspond to empty cut-sets having all nodes in exactly one of these sets, i.e., we have $\StateProcessed{S_{\emptyset}} = \emptyset$ and $\StateProcessed{S_V} = V$, and, symmetrically, $\StateNotProcessed{S_{\emptyset}} = V$ and $\StateNotProcessed{S_V} = \emptyset$.
Having defined the set of nodes of $N$, we now turn our attention to its arcs.
On an intuitive level, two nodes $S_X = (F_X, \prec_{F_X}, \sigma_{F_X})$ and $S_Y = (F_Y, \prec_{F_Y}, \sigma_{F_Y})$ are connected by an arc $S_XS_Y \in A$ if $\DirectedCut{F_Y}$ can be obtained from $\DirectedCut{F_X}$ by moving one vertex from $\StateNotProcessed{S_X}$ to $\StateProcessed{S_X}$.
More formally, we have $S_XS_Y \in A$ if and only if all the following criteria are met:
\begin{enumerate}[($\mathcal{A}$C1)]
	\item There exists a vertex $v \in \StateNotProcessed{S_X}$ such that\label[AC]{ac:directed-cut-set}
	\begin{align*}
		F_Y = (F_X \setminus \{uv \in F_X \mid u \in \StateProcessed{S_X}\}) \cup \{vw \mid w \in \StateNotProcessed{S_X}\}.
	\end{align*}
	This implies that $v \in \StateProcessed{S_Y}$ holds.
	Furthermore, the vertex $v$ is unique as $G$ is connected.
	Thus, we \emph{label} the edge $S_XS_Y$ with $v$ and write $\lambda(S_XS_Y) = v$.
	\item For all $u,w \in V(F_X \cap F_Y)$ we have $u \prec_{F_X} w \Leftrightarrow u \prec_{F_Y} w$ and for all $e \in F_X \cap F_Y$ we have $\sigma_{F_X}(e) = \sigma_{F_Y}(e)$.\label[AC]{ac:consistency}
	\item If $\lambda(S_XS_Y)$ is an endpoint of an edge in $F_X$, it must be the leftmost sink in $\prec_{F_X}$. If $\lambda(S_XS_Y)$ is an endpoint of an edge in $F_Y$, it must be the rightmost source in $\prec_{F_Y}$.
    \label[AC]{ac:leftmost-rightmost}
	\item The linear layouts $\linL[F_X]$ and $\linL[F_Y]$ are valid $\ell$-page $q$-width stack (or queue) layouts of the graphs $G[F_X]$ and $G[F_Y]$, respectively.\label[AC]{ac:transition}
\end{enumerate}

Note that \Cref{ac:directed-cut-set,ac:consistency} ensure consistency for $\prec_{F_Y}$ and $\sigma_{F_Y}$ with respect to $\prec_{F_X}$ and $\sigma_{F_X}$.
That is, in $S_Y$ we move exactly the vertex $\lambda(S_XS_Y)$ from $\StateNotProcessed{S_X}$ to $\StateProcessed{S_X}$ but do not make further changes to $\prec_{F_Y}$ and $\sigma_{F_Y}$.
In the end, we will use \Cref{ac:leftmost-rightmost,ac:transition} to show that the obtained layout is indeed a stack or queue layout.
This completes the definition of the state graph $H = (N, A)$; we schematize an intuitive example in \Cref{fig:poly-pw-state-graph}.
\end{statelater}

\shortLong{
	\subsection{Using Nicely Oriented Cut-Sets to Show \Cref{thm:pagewidth}}\label{sec:poly-state-graph-short}
	The state graph $H$ contains a node for each possible nicely-oriented cut-set
	of $G$, together with the corresponding page assignment.
	The arcs in $H$ represent possible transitions from one nicely oriented cut-set to another that can occur when traversing a hypothetical solution from left to right (see \Cref{fig:poly-pw-state-graph}), i.e., that are obtained by moving exactly one vertex from the right to the left.
	A state $S = (F, \prec_{F_X}, \sigma_{F_X})$ is \emph{feasible} if and only if there exists an $\ell$-page $q$-width stack layout \linL[S] of $G[\StateProcessed{S} \cup V(F)]$ such that for all $u,v \in V(F)$ we have $u \prec_{F} v \Leftrightarrow u \prec_S v$ and $\sigma_F = \sigma_S\vert_{F}$.
	For queue layouts, the definition is analogous.
	In the full version, we show that, essentially, having arcs only between state pairs where both induce valid stack (or queue) layouts on their respective cut-sets is enough to preserve feasibilty along arcs.
}{
	 Next, we introduce the concept of feasible states for stack layouts:
	A state $S = (F, \prec_{F_X}, \sigma_{F_X})$ is \emph{feasible} if and only if there exists an $\ell$-page $q$-width stack layout \stackL[S] of $G[\StateProcessed{S} \cup V(F)]$ such that for all $u,v \in V(F)$ we have $u \prec_{F} v \Leftrightarrow u \prec_S v$ and $\sigma_F = \sigma_S\vert_{F}$.
	For queue layouts, the definition is analogous.
	We show that the arcs of our state graph preserve feasible states.
}

\begin{figure}
	\centering
	\includegraphics[page=1]{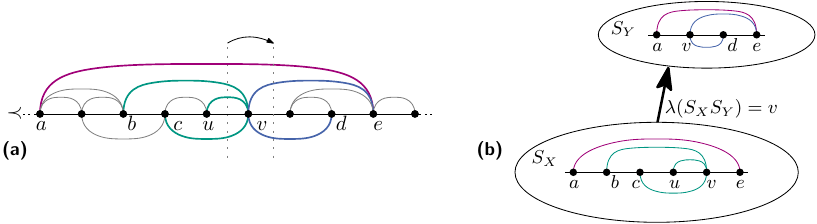}
	\caption{\textbf{\textsf{(a)}} Parts of a stack layout of $G$ with the two cut-sets $\EdgesCrossingV[\prec]{u}$ and $\EdgesCrossingV[\prec]{v}$. \textbf{\textsf{(b)}} The corresponding nodes in the state graph. The arc in~\textsf{(b)} corresponds to the transition from $\EdgesCrossingV[\prec]{u}$ to $\EdgesCrossingV[\prec]{v}$ in~\textsf{(a)}. Edges appearing in at least one cut-set are colored.}
	\label{fig:poly-pw-state-graph}
\end{figure}

\begin{restatable}\restateref{lem:poly-pw-feasible-states}{lemma}{lemmaPolyPwFeasibleStates}
	\label{lem:poly-pw-feasible-states}
	Let $H = (N, A)$ be a state graph and $S_XS_Y \in A$ one of its arcs. If $S_X$ is feasible for stack (queue) layouts then so is $S_Y$.
\end{restatable}
\begin{prooflater}{plemmaPolyPwFeasibleStates}
	We first show the statement for stack layouts and afterwards turn our attention to queue layouts.
	Throughout the proof, we let $S_X = (F_X, \prec_{F_X}, \sigma_{F_X})$ and $S_Y = (F_Y, \prec_{F_Y}, \sigma_{F_Y})$ be two nodes of $H$ with $S_XS_Y \in A$.

	\proofsubparagraph*{Stack Layouts.}
	Assume that $S_X$ is a feasible state, i.e., there exists an $\ell$-page $q$-width stack layout $\stackL[G_X]$ for the graph $G_X = G[\StateProcessed{S_X} \cup V(F_X)]$.
	We now construct a stack layout $\stackL[G_Y]$ for the graph $G_Y = G[\StateProcessed{S_Y} \cup V(F_Y)]$ based on $\stackL[G_X]$.
	For that, it is useful to observe that $G_X$ and $G_Y$ differ at most in terms of the vertex $\lambda(S_XS_Y)$ and its incident edges.
	To construct the spine order $\prec_{G_Y}$, we start by copying $\prec_{G_X}$.
	If $\lambda(S_XS_Y) \notin V(G_X)$, we insert it in $\prec_{G_Y}$ by setting $\lambda(S_XS_Y) \prec_{G_Y} w$ for all $w \in V(G_Y)$ which are a sink in $\DirectedCut{F_Y}$ and $u \prec_{G_Y} \lambda(S_XS_Y)$ for all other vertices $\lambda(S_XS_Y) \neq u \in V(G_Y)$.
	Finally, we order the endpoints of all edges in $F_Y$ incident to $v$ as specified in $\prec_{F_Y}$ and obtain $\prec_{G_Y}$ by taking the transitive closure of the above defined partial orders.
	For the page assignment $\sigma_{G_Y}$, we set $\sigma_{G_Y}(e) = \sigma_{F_Y}(e)$ for all $e \in F_Y$ and $\sigma_{G_Y}(e) = \sigma_{G_X}(e)$ otherwise. %

	To show that $\stackL[G_Y]$ witnesses the feasibility of $S_Y$, we have to show that it is a valid $\ell$-page $q$-width stack layout of $G_Y$ with $\prec_{F_Y} \subseteq \prec_{G_Y}$ and $\sigma_{F_Y} = \sigma_{G_Y}\vert_{F}$.
	To see that the latter two criteria hold, observe that our construction mimics, on the one hand, $\linL[F_Y]$ for all endpoints of edges in $F_Y \setminus F_X$, and, on the other hand, $\stackL[G_X]$ (which witnesses feasibility of $S_X$) for all remaining edges in $F_X \cap F_Y$ and their endpoints.
	As \Cref{ac:consistency} ensures that these edges and their endpoints have in $S_X$ and $S_Y$ the same page assignment and relative order, we conclude that $\prec_{F_Y} \subseteq \prec_{G_Y}$ and $\sigma_{F_Y} = \sigma_{G_Y}\vert_{F}$ indeed holds.

	For the former criterion, we first observe that the number of pages follows directly from the construction.
	To show that $\pw[\stackL[G_Y]] \leq q$, we assume towards a contradiction that this is not the case.
	Hence, there exists a vertex $v \in V(G_Y)$ and a page $p \in [\ell]$ such that we have $\Size{\EdgesCrossingVPage{v}{p}} > q$.
	There are two cases that we can consider depending on the relative position of $v$ and $\lambda(S_XS_Y)$:
	The first case, $\lambda(S_XS_Y) \preceq v$ immediately yields a contradiction, as we have in this case $\EdgesCrossingV{v} \subseteq F_Y$.
	Thus, $\Size{\EdgesCrossingVPage{v}{p}} > q$ would contradict the fact that $S_Y$ is a feasible state; recall \Cref{nc:pagewidth}.
	This leaves us only with the case $v \prec \lambda(S_XS_Y)$.
	However, since we copy $\stackL[G_X]$ to construct $\stackL[G_Y]$, all edges in $\EdgesCrossingVPage{v}{p}$, including those incident to $\lambda(S_XS_Y)$, are present in $\stackL[G_X]$ on the same pages and with the same relative order among their endpoints; recall also \Cref{ac:consistency}.
	Hence, $\pw[\stackL[G_X]] > q$ holds, a contradiction to the fact that $S_X$ is feasible.
	Thus, we conclude $\pw[\stackL[G_Y]] \leq q$.

	We now prove that $\stackL[G_Y]$ is also a valid stack layout.
	For the sake of a contradiction, assume that there are two edges $u_1w_1, u_2w_2 \in E(G_Y)$ such that $\sigma_{G_Y}(u_1w_1) = \sigma_{G_Y}(u_1w_1)$ and $u_1 \prec_{G_Y} u_2 \prec_{G_Y} w_1 \prec_{G_Y} w_2$, i.e., they cross on page $\sigma_{G_Y}(u_1w_1)$.
	First, recall that the graph $G_Y$ differs from $G_X$ only in terms of $\lambda(S_XS_Y)$ and its incident edges and that $\stackL[G_Y]$ is constructed based on $\stackL[G_X]$.
	Thus, a crossing in $\stackL[G_Y]$ between $u_1w_1$ and $u_2w_2$ implies that at least one edge is part of $F_Y \setminus F_X$.
	Otherwise, $\stackL[G_X]$ would not be a valid stack layout, which contradicts the feasibility of $S_X$.
	Having $u_1w_1 \in F_Y \setminus F_X$ implies $u_1 = \lambda(S_XS_Y)$, which contradicts $u_2w_2 \in E(G_Y)$ as we have $u_1 \prec_{G_Y} u_2 \prec_{G_Y} w_1 \prec_{G_Y} w_2$, i.e., $u_2, w_2 \in \StateNotProcessed{S_Y}$.
	Thus, $u_2w_2 \in F_Y \setminus F_X$ and $u_2 = \lambda(S_XS_Y)$ holds.
	However, from $u_1 \prec_{G_Y} u_2 \prec_{G_Y} w_1$ and the definition of $S_Y$ it follows $u_1w_1 \in F_Y$.
	However, a crossing between two edges from $F_Y$ contradicts the existence of the edge $S_XS_Y \in \mathcal{A}$, as $\linL[F_Y]$ must be a valid stack layout; recall also \Cref{ac:transition}.
	Combining all, we conclude that $\stackL[G_Y]$ is a valid $\ell$-page $q$-width layout that witnesses the feasibility of $S_Y$.

	\proofsubparagraph*{Queue Layouts.}
	The proof for queue layouts is analogous to the one above for stack layouts.
	In particular, observe that all arguments except those showing the non-existence of a crossing purely exploit general properties of linear layouts.
	Only for the last argument, we worked under the assumption that two edges on a page cross.
	For queue layouts, let us now assume that two edges $u_1w_1, u_2w_2 \in E(G_Y)$ on the page $\sigma_{G_Y}(u_1w_1) = \sigma_{G_Y}(u_1w_1)$ nest, i.e., that we have $u_1 \prec_{G_Y} u_2 \prec_{G_Y} w_2 \prec_{G_Y} w_1$.
	However, by the very same arguments we conclude that at least one edge must be in $F_Y \setminus F_X$ and that this edge must be $u_2w_2$.
	This allows us to again deduce from $u_1 \prec_{G_Y} u_2 \prec_{G_Y} w_2 \prec_{G_Y} w_1$ that $u_1w_1 \in F_Y$ holds.
	As this implies that the nesting already exists in $\linL[F_Y]$, this contradicts the arc $S_XS_Y$ due to criterion \Cref{ac:transition}.
	Hence, we also conclude for queue layouts that the feasibility of $S_X$ and the arc $S_XS_Y\in \mathcal{A}$ implies the feasibility of $S_Y$.
\end{prooflater}

We now have all tools at hand to present our algorithm.
\shortLong{
	Since the arcs of $H$ preserve existence of a (partial) $\ell$-page $q$-width stack (or queue) layout of $G$,
}{%
	In particular, observe that our state graph $H$ contains a node for each possible nicely-oriented cut-set that is induced by an $\ell$-page $q$-width stack (or queue) layout of $G$ (if it exists) together with the page assignment of its edges.
	Since the arcs of $H$ preserve existence of a (partial) solution,%
} we have reduced the problem of finding such a layout to the problem of finding a path in $H$ between the two special nodes $S_{\emptyset}$ and $S_V$.
\shortLong{In the full version, we}{We} show that $H$ can be computed in polynomial time for constant $\ell$ and $q$ and use this to summarize the main result of this section in the following theorem.

\thmpagewidth*
\begin{prooflater}{pthmpagewidth}
	Let $G$ be a graph on $n$ vertices and $m$ edges.
	First, we check if $m \leq (\ell + 1)n - 3\ell$ holds for stack, or $m \leq 2\ell n-\ell(2\ell + 1)$ for queue layouts and reject otherwise, as this are known upper bounds for the number of edges of a $\ell$-page stack and queue layout~\cite{dujmovic2004linear}.

	We now construct the state graph $H = (N, A)$ for $G$, $\ell$, and $q$.
	First, we bound its size and observe that $G$ has $\binom{m}{q\cdot\ell}$ different cut-sets $F \subseteq E$ of size $1 \leq \Size{F} \leq q\cdot\ell$ that we consider.
	Let $F \subseteq E$ be one such cut-set and recall that it consists of at most $q\cdot\ell$ edges.
	There are at most $\ell^{q \cdot \ell}$ different page assignment $\sigma_{F}$ of the edges $F$ and $(2\cdot q\cdot\ell)!$ different linear orders $\prec_{F}$ of their endpoints.
	Since $H$ has (at most) one node for each combination of $F$, $\prec_{F}$, and $\sigma_F$, this leaves us with at most $\BigO{m^{q \cdot \ell} \cdot \ell^{q \cdot \ell} \cdot (2\cdot q\cdot\ell)!}$ nodes overall, including the two dummy nodes $S_{\emptyset}$ and $S_{V}$.
    Recall $m = \BigO{n \cdot \ell}$.
    Since we can assume $q \leq m$ and $\ell \leq n^2$, and we have $\BigO{(2\cdot q\cdot\ell)!} = \BigO{(q\cdot\ell)^{q \cdot \ell}}$, we can bound each term and their product with $n^{\BigO{q \cdot \ell}}$.
    This is, therefore, a bound on the number of nodes in $H$.
    To count the number of arcs in $H$, we bound the outdegree of each node $S_X$ and observe that for each choice of $v \in \StateNotProcessed{S}$ the cut-set $F_Y$ obtained from $S_X$ is unique.
	Thus, the outdegree of each node is in $\BigO{n \cdot \ell^{q \cdot \ell} \cdot (2\cdot q\cdot\ell)!}$.
    By similar arguments, we can also bound this by $n^{\BigO{q \cdot \ell}}$.
    As each of the $n^{\BigO{q \cdot \ell}}$ nodes can have at most $n^{\BigO{q \cdot \ell}}$ outgoing edges, also the overall size of $H$ is in $n^{\BigO{q \cdot \ell}}$.
    Since the criteria \Cref{ac:directed-cut-set}--\ref{ac:transition} can be checked in polynomial time for each of the $n^{\BigO{q \cdot \ell}}$ pairs of nodes in $H$, we can construct the state graph $H$ in $n^{\BigO{q \cdot \ell}}$ time.

	Once the state graph $H$ has been constructed, we check if $H$ contains a directed path $P$ from $S_{\emptyset}$ to $S_V$.
	If so, we report that $G$ has a $\ell$-page $q$-width stack (or queue) layout.
	To also output a solution, we can obtain the spine order by concatenating the labels of the arcs in the order visited by $P$.
	As $G$ is connected and $H$ is easily seen to be acyclic (for every arc $S_XS_Y \in A$ we have $\StateProcessed{S_X} \subset \StateProcessed{S_Y}$), we reach $S_V$ only after having seen each vertex $v \in V$ exactly once as the label of an arc.
	The page assignment can be obtained by merging the individual page assignments stored in the nodes that form $P$.
	\Cref{ac:consistency} ensures that they are consistent.
	If there is no such path $P$ in $H$, we report that no desired layout exists.
	Note that we can use a BFS in $H$ with $S_{\emptyset}$ as source node to find $P$ in $n^{\BigO{q \cdot \ell}}$ time.

	\proofsubparagraph*{Correctness of the Algorithm.}
	To argue correctness of the algorithm, we need to show that the path $P$ exists if and only if $G$ has an $\ell$-page $q$-width stack (or queue) layout.
	For the forward direction, i.e., that the existence of a path $P$ implies that $G$ possesses a solution, we observe that the node $S_{\emptyset}$, i.e., the node that corresponds to the empty cut-set with $\StateProcessed{S_{\emptyset}} = \emptyset$, is feasible as $G[\StateProcessed{S_{\emptyset}}]$ is the empty graph and $S_{\emptyset}$ has no cut-set edges.
	Hence, by inductively applying \Cref{lem:poly-pw-feasible-states} on the nodes of $P$, we can eventually conclude that also the state $S_{V}$ is feasible.
	As $\StateProcessed{S_V} = V$, we have $G[\StateProcessed{S_V}] = G$, i.e., the layout that witnesses the feasibility also serves as a solution.
	Furthermore, it corresponds to the layout obtained from the path as described above.

	For the backward direction, i.e., that a solution $\linL$ of $G$ implies the existence of a path $P$ between $S_{\emptyset}$ and $S_V$ in $H$, we use \Cref{lem:poly-pw-cuts} to conclude that for each $v \in V$ there exists exactly one node $S_v \in N$ that contains precisely the edges in $\EdgesCrossingV{v}$ and matches the page assignment and vertex order from $\linL$.
	Thus, to show that $P$ exists in $H$, it remains to show that the required arcs exist in $H$.
	To see that $S_uS_v \in \mathcal{A}$ holds for two neighboring vertices $u,v\in V(G)$ with $u \prec v$, we can check each of the criteria \Cref{ac:directed-cut-set,ac:consistency,ac:leftmost-rightmost,ac:transition} and see that they hold:
	\Cref{ac:directed-cut-set} mimics the construction of $\EdgesCrossingV{v}$ based on $\EdgesCrossingV{u}$.
	In particular, this removes from $\EdgesCrossingV{u}$ all edges incident to $v$ that have as endpoint a vertex $w \in \LeftOfV{v}$ and adds those with an endpoint $w \in \RightOfV{v}$.
	Note that therefore also $\lambda(S_uS_v) = v$ holds as one would expect.
	\Cref{ac:consistency,ac:leftmost-rightmost} holds by construction and the existence of \linL witnesses that also Criterion \Cref{ac:transition} holds.
	Finally, applying the same reasoning allows us also to conclude that $S_{\emptyset}S_u, S_vS_V \in A$ holds, where $u$ and $v$ are the leftmost and rightmost vertices in $\prec$, respectively.
	Combining all, we conclude that the path $P$ must exist in $H$.
\end{prooflater}

The algorithm from \Cref{thm:pagewidth} runs in polynomial time for constant $\ell$ and $q$.
Furthermore, recall that deciding if a graph has a 2-page stack layout or 1-page queue layout is \NP-complete~\cite{yannakakis1989embedding,heath1992laying}.
An $\ell$-page $1$-width stack or queue layout witnesses that $G$ has cutwidth at most $\ell$, which is in general \NP-complete to decide~\cite{GJ.CIG.1979}.
This means that under well-established complexity assumptions, it is not possible to preserve \XP-tractability if either of the two parameters is dropped.

\section{Single-Exponential Algorithm for 1-Page Queue Layouts}
\label{sec:single-exp-ql}
In this section, we show that we can find a 1-page queue layout, if it exists, in $2^{\BigO{n}}$ time, which improves on the trivial algorithm that exhaustively considers all $n!$ possible spine orders.
Our algorithm exploits the equivalence between 1-page queue layouts and arched leveled planar embeddings established by Heath and Rosenberg~\cite{heath1992laying}.
We first show that we can reduce the problem of deciding whether a graph admits a 1-page queue layout to deciding whether at least one of $2^{\BigO{n}}$ \emph{labeled} instances of \probname{Arched Leveled Planarity} admits a solution.
While \probname{Arched Leveled Planarity} is \NP-complete~\cite{heath1992laying}, our labeled instances can be solved in linear time thanks to a reduction to the well-studied problem \probname{Level Planarity}~\cite{JL.LPE.1999,JL.LPE.2002}. %
We will assume that $G$ is connected (as each component can be treated separately) and, as $G$ must be planar~\cite{heath1992laying}, we assume $\Size{E} = \BigO{n}$.

We translate the definition of Heath and Rosenberg~\cite{heath1992laying} into modern terms as follows.
A graph $G$ has a \emph{leveled planar embedding} if there exists a \emph{level assignment} function $\gamma\colon V \to [h]$ that maps each $v \in V$ to one of $h \geq 1$ \emph{levels}, and a straight-line planar drawing $\Gamma$ of $G$ such that each vertex $v$ has y-coordinate $\gamma(v)$ and all edges are \emph{proper}, i.e., for all $uv \in E$ we have $\Size{\gamma(u) - \gamma(v)} = 1$.
The latter property allows us to express the drawing as a collection of total orders $\prec_i$, one for each level $i \in [h]$, that specifies the left-to-right order of the vertices on level $i$~\cite{JL.LPE.2002}.
For each level $i$, let $s_i$ be the rightmost vertex on level $i$ in $\Gamma$ such that we have $s_iw \in E$ for some $w \in V$ with $\gamma(w) = i + 1$, or if there is no such vertex, let $s_i$ be the leftmost vertex on level $i$.
\emph{Arched leveled planar embeddings} extend leveled planar embeddings by also allowing non-monotone  \emph{arching} edges $t_iw$ between the leftmost vertex $t_i$ on level $i$ and vertices $w \in V$ with $\gamma(w) = i$ and $s_i \preceq_i w$~\cite{heath1992laying}.
These edges can be drawn crossing-free below the lowest level; see \Cref{fig:single-exp-ql-arched}b.
\onlyLong{
The problem \probname{Arched Level Planarity} asks for a given graph $G$ together with a level assignment $\gamma$, if $G$ admits an arched leveled planar embedding with the level assignment $\gamma$.
In the problem \probname{Arched Leveled Planarity} studied by Heath and Rosenberg, we are only given $G$ and ask for an arched leveled planar embedding of it, i.e., the level assignment is not specified in the input~\cite{heath1992laying}.
At this point, we would like to make the reader aware of the closely related, well-studied, and linear-time solvable problem \probname{Level Planarity}~\cite{JL.LPE.1999,JL.LPE.2002}.
There, we seek a planar $y$-monotone drawing of $G$ where the $y$-coordinate of each vertex is specified by the given level assignment $\gamma$.
Note that there also edges between non-adjacent layers are allowed.
We will come back to \probname{Level Planarity} later.

}
\onlyShort{The problem \probname{Arched Leveled Planarity} asks whether $G$ has a level assignment that admits an arched leveled planar embedding.}
Heath and Rosenberg~\cite{heath1992laying} showed that the spine order of a 1-page queue layout of a graph $G$ induces a level assignment together with a vertex order on each level that together yield an arched leveled planar embedding of $G$ and vice versa; see \Cref{fig:single-exp-ql-arched}.

\begin{figure}
	\centering
	\includegraphics[page=1]{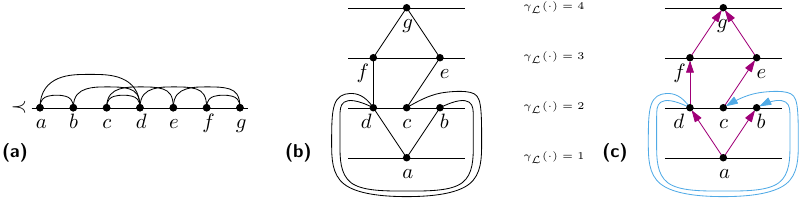}
	\caption{\textbf{\textsf{(a)}} A 1-page queue layout of $G$ and~\textbf{\textsf{(b)}} the corresponding arched leveled planar embedding obtained by the equivalence established by Heath and Rosenberg~\cite{heath1992laying}. \textbf{\textsf{(c)}} The labeling induced by the embedding from~\textbf{\textsf{(b)}}: Arcs labeled $\LabelOrdindary$ and $\LabelArch$ are colored lilac and blue, respectively.}
	\label{fig:single-exp-ql-arched}
\end{figure}

We show that, when seeking a 1-page queue layout of $G$, one can branch over the information needed to infer the level assignment for a corresponding arched leveled planar embedding. 
Consider an arbitrary orientation of the edge set $E$ as an oriented arc set $A$ and let $\delta\colon A \to \LabelImage$ be a function that determines if an arc $uv \in A$ should be \emph{ordinary} (\LabelOrdindary) or \emph{arching} (\LabelArch). %
The tuple $\mathcal{L}=(A, \delta)$ is called a \emph{labeling} of $G$; see also \Cref{fig:single-exp-ql-arched}c.
A level assignment $\gamma$ is said to be \emph{consistent} with $\mathcal{L}$ if and only if, for every $uv \in A$, it holds $\gamma(u) = \gamma(v) - 1$ if $\delta(uv) =\ \LabelOrdindary$ and $\gamma(u) = \gamma(v)$ if $\delta(uv) = \LabelArch$.
Similarly, an embedding $\Gamma$ is consistent with $\mathcal{L}$ if its level assignment is consistent and we have $u \prec_{\gamma(u)} v$ for every $uv \in \delta^{-1}(\LabelArch)$.
Clearly, the level assignment $\gamma_{\Gamma}$ of every arched leveled planar embedding $\Gamma$ of $G$ induces a labeling $(A_{\Gamma}, \delta_{\Gamma})$ that is consistent with $\Gamma$.
Conversely, we can reobtain $\gamma_{\Gamma}$ from this $(A_{\Gamma}, \delta_{\Gamma})$ by fixing an arbitrary vertex to level 1, performing a BFS from $v$ using the direction information in $A_{\Gamma}$ to determine levels for all other vertices above and below, and potentially shifting all assigned levels such that their numbers span the same range as those in $\Gamma$.
This is formalized by the following lemma.

\begin{restatable}\restateref{lem:single-exp-ql-edge-labeling}{lemma}{lemmaSingleExpEdgeLabeling}
	\label{lem:single-exp-ql-edge-labeling}
	Let $\mathcal{L} = (A, \delta)$ be a labeling of $G$.
	In linear time we can compute a level assignment $\gamma_{\mathcal{L}}\colon V \to [h]$ that is consistent with the labeling or report that no such assignment exists.
	If there exists an arched leveled planar embedding $\Gamma$ that induces $\mathcal{L}$, then there also exists one that induces $\mathcal{L}$ and uses the level assignment $\gamma_{\mathcal{L}}$.
\end{restatable}
\begin{prooflater}{plemmaSingleExpEdgeLabeling}
	We first discuss how we can derive a level-assignment $\gamma_{\mathcal{L}}$ that is consistent with $\mathcal{L}$, if one exists.
    In a second step, we show that an embedding of $G$ can be transformed into one that uses $\gamma_{\mathcal{L}}$.
	
	\proofsubparagraph*{Computing $\boldsymbol{\gamma_{\mathcal{L}}}$.}
	We perform a breadth-first search (BFS) traversal on the (undirected) graph $G$, starting at an arbitrary vertex $v \in V$ and initialize our level-assignment by setting $\gamma_{\mathcal{L}}(v) = 1$.
	Whenever we visit a vertex $w \in V$ via the edge $uw \in E$, we check the labeling $\mathcal{L}$ and extend the level assignment as follows:
    If $\delta(uw) =\ \LabelOrdindary$, we set $\gamma_{\mathcal{L}}(w) = \gamma_{\mathcal{L}}(u) + 1$ if $uw \in A$ or $\gamma_{\mathcal{L}}(w) = \gamma_{\mathcal{L}}(u) - 1$ if $wu \in A$.
    Otherwise, i.e., if $\delta(uw) = \LabelArch$, we set $\gamma_{\mathcal{L}}(w) = \gamma_{\mathcal{L}}(u)$.
    Before we assign a level to $w$, we check if it has already been assigned one.
    If so, we verify that the existing level assignment is consistent with the one determined by the above rules.
    If it is not, we report that no level assignment consistent with $\mathcal{L}$ exists.
    As $G$ is connected, every vertex gets assigned a level eventually (unless we report that no consistent level assignment exists).
    Depending on the choice of $v$, we can have $\gamma_{\mathcal{L}}(u) < 1$ for some $u \in V$.
    However, we can add a constant offset to ensure that vertices are only assigned to positive levels.
	Finally, we observe that $\gamma_{\mathcal{L}}$ is consistent with $\mathcal{L}$ by construction and can be computed in \BigO{n} time.

	\proofsubparagraph*{Equivalence with the Level-Assignment $\boldsymbol{\gamma_{\mathcal{L}}}$.}
	Let $\Gamma$ be an arched leveled planar embedding of $G$ that induces $\mathcal{L}$.
    By the same arguments as above, we can add (or remove) a constant $y$-offset to every vertex to obtain an embedding that starts at level one.
    Therefore, we can assume that the level assignment $\gamma_{\Gamma}$ of $\Gamma$ places vertices on every level in $[h]$ for some $h \geq 1$.
    As $\Gamma$ induces $\mathcal{L}$, $\gamma_{\Gamma}$ witnesses the consistency with $\mathcal{L}$, and our procedure described above returns a consistent level assignment $\gamma_{\mathcal{L}}$.
	Towards a contradiction, assume that $\gamma_{\Gamma}$ and $\gamma_{\mathcal{L}}$ are not equal.
	Thus, there exists a vertex $v \in V$ such that $\gamma_{\Gamma}(v) \neq \gamma_{\mathcal{L}}(v)$.
	We assume without loss of generality $\gamma_{\Gamma}(v) > \gamma_{\mathcal{L}}(v)$ as the other case is symmetric.
    Let $\gamma_{\Gamma}(v) - \gamma_{\mathcal{L}}(v) = d > 0$.
    As $\gamma_{\Gamma}(v) > \gamma_{\mathcal{L}}(v) \geq 1$, there exists a vertex $u \in V$ with $\gamma_{\Gamma}(u) = 1$.
    As $G$ is connected, there is a path $P$ between $v$ and $u$.
    Since both $\gamma_{\Gamma}$ and $\gamma_{\mathcal{L}}$ are consistent with $\mathcal{L}$, for every vertex $w$ on $P$ we have $\gamma_{\Gamma}(w) - \gamma_{\mathcal{L}}(w) = d$.
    In particular, $\gamma_{\Gamma}(u) - \gamma_{\mathcal{L}}(u) = d$.
    Since $\gamma_{\Gamma}(u) = 1$, we obtain from above equality $\gamma_{\mathcal{L}}(u) \leq 0$.
    This is a contradiction to the fact that $\gamma_{\mathcal{L}}(w) > 0$ holds for every vertex $w\in V$.    
	Hence, $\gamma_{\Gamma}(v) \neq \gamma_{\mathcal{L}}(v)$ cannot hold and we conclude %
	$\gamma_{\Gamma} = \gamma_{\mathcal{L}}$.
\end{prooflater}

\onlyShort{To now decide if a given graph $G$ has an arched leveled planar embedding $\Gamma$, we can branch over all $\BigO{4^n} = 2^{\BigO{n}}$ different labelings $\mathcal{L} = (A, \delta)$ of $G$.
In each branch, we use \Cref{lem:single-exp-ql-edge-labeling} to compute a level assignment $\gamma_{\mathcal{L}}$ or immediately reject the branch.
Thus, it remains to check in each branch, i.e., for each level assignment, if $G$ admits a corresponding arched leveled planar embedding. %
We do so using a linear-time reduction to the problem \probname{Level Planarity}, which requires a level assignment as part of its input but does not allow arching edges.
We emulate arching edges in \probname{Level Planarity} as follows (see also \Cref{fig:single-exp-ql-reduction}).
First, we enclose the whole graph $G$ in a \emph{frame}, which is a cycle spanning from below the lowest level of $G$ to above its highest level that will have a unique embedding (up to reflection) and for which we ensure that the (remaining) graph $G$ is drawn inside of it.
Between each pair of levels $i, i+1$, we add a level $i+0.5$ and subdivide all edges to ensure properness.
We use $l_i$ and $r_i$ to refer to the two vertices of the frame cycle that are on level $i+0.5$.
We replace each arching arc $uv \in A$ with $\delta(uv) =\ \LabelArch$ on level $i$ with the edges $ul_{i-1}$, $ul_{i}$, and $vr_{i}$ to $G'$; see \Cref{fig:single-exp-ql-reduction}.
Observe that this forces $u$ to become the leftmost vertex on level $i$ and $v$ to be right of any vertex adjacent to one on level $i + 1$.
We provide a formal construction and show equivalence in \ifthenelse{\boolean{cameraready}}{the full version~\cite{ARXIV}}{\Cref{app:sec:single-exp-ql}}.
Using the known algorithm for testing \probname{Level Planarity} in linear time~\cite{JL.LPE.1999,JL.LPE.2002}, we finally obtain our third contribution.}

\begin{figure}
	\centering
	\includegraphics[page=1]{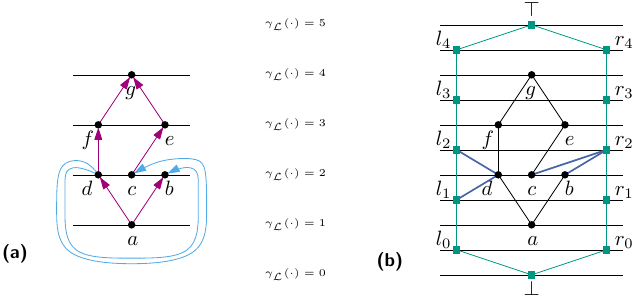}
	\caption{\textbf{\textsf{(a)}} An instance of \probname{Labeled Arched Level Planarity} and~\textbf{\textsf{(b)}} the obtained instance of \probname{Level Planarity}. We also visualize a drawing and color the frame and edges that simulate the arching edges from~\textbf{\textsf{(a)}} in green and blue, respectively. Only relevant subdivision vertices are shown.}
	\label{fig:single-exp-ql-reduction}
\end{figure}

\begin{statelater}{constArchLevelPlan}
Let $(G, \mathcal{L})$ be an instance of \probname{Arched Leveled Planarity}.
We use \Cref{lem:single-exp-ql-edge-labeling} to obtain a level assignment $\gamma_{\mathcal{L}}$ that is consistent with $\mathcal{L}$.
If no such assignment exists, we immediately return a negative instance.
Otherwise, we let $h = \max_{v \in V} \gamma_{\mathcal{L}}(v)$.
In our reduction, we introduce for each level $i \in [h]$ a dummy level $i + 0.5$ and, furthermore, three dummy levels $0$, $0.5$, and $h + 1$.
We now construct a graph $G'$ and initially set $V(G') = V(G)$ and $E(G') = \{uv \in A \mid \delta(uv) =\ \LabelOrdindary\}$.
Next, we initialize $\gamma' = \gamma_{\mathcal{L}}$ and create for each level $i \in [h] \cup \{0\}$ two vertices $l_{i}$ and $r_{i}$, and set $\gamma'(l_{i}) = \gamma'(r_{i}) = i + 0.5$.
In addition, we create the vertices $\bot$ and $\top$, and set $\gamma'(\bot) = 0$ and $\gamma'(\top) = h + 1$.
We connect these vertices in a cycle, i.e., have the edges $l_il_{i+1}$ and $r_ir_{i + 1}$ for each $i \in [h - 1] \cup \{0\}$ as well as $\bot l_0$, $\bot r_0$, $l_h\top$, and $r_h\top$, and refer to it as the \emph{frame}.
Note that the frame has a unique level planar drawing up to reflection.
Conceptually, we use the frame to route the arching edges around our drawing.
To do that, we consider each level $i \in [h]$ and every arc $uv \in A$ with $\delta(uv) =\ \LabelArch$ and $\gamma_{\mathcal{L}}(u) = i$, i.e., that arches on level $i$ and 
add the edges $ul_{i-1}$, $ul_{i}$, and $vr_{i}$ to $G'$.
This forces $u$ to become the leftmost vertex on level $i$ and $v$ to be right of any vertex adjacent to one on level $i + 1$.
If we connect two different vertices on the same level to the left side of the frame, we directly return a negative instance as in this case no arched leveled planar drawing can induce $\mathcal{L}$.
We obtain an instance $(G', \gamma')$ of \probname{Level Planarity} and with the following lemma, we establish the correctness of the reduction.

\begin{restatable}\restateref{lem:single-exp-ql-reduction-level-planarity}{lemma}{lemmaSingleExpReductionLP}
	\label{lem:single-exp-ql-reduction-level-planarity}
	The instance $(G, \mathcal{L})$ has an arched level planar embedding $\Gamma$ that induces $\mathcal{L}$ if and only if the instance $(G', \gamma')$ of \probname{Level Planarity} has a level planar drawing $\Gamma'$.
\end{restatable}
\begin{prooflater}{plemmaSingleExpReductionLP}
	We show both directions separately. 
	Note that if no edge is labeled as arching edge $(G, \mathcal{L})$, both instances are, besides the frame, identical.
	Since the frame can always be removed from a drawing of $(G', \gamma')$ or drawn next to an embedding of $(G, \mathcal{L})$, we focus on the more interesting case that includes arching edges.
	There, we use the frame to argue that arching edges are correctly represented in $(G', \gamma')$.
	Let $h$ be the number of levels in $\gamma_{\mathcal{L}}$.
	
	\proofsubparagraph*{($\boldsymbol{\Rightarrow}$)}
	Let $\Gamma$ be an arched level planar embedding that induces $\mathcal{L}$.
	We now construct based on $\Gamma$ a drawing $\Gamma'$ for $(G', \gamma')$.
	First, note that we have $\gamma'(v) = \gamma_{\mathcal{L}}(v)$ for every vertex $v \in V(G)$.
	We thus use in $\Gamma'$ on every level $i \in [h]$ the same vertex order as in $\Gamma$.
	Next, we draw in $\Gamma'$ the frame as shown in green in \Cref{fig:single-exp-ql-reduction}b, i.e., for each $i \in [h] \cup \{0\}$, we set $l_i \prec_{i + 0.5} r_i$.
	The remaining two vertices $\bot$ and $\top$ are the only ones on their respective level.
	We now have a total order among all vertices on the same level.
	The embedding $\Gamma$ guarantees that these orders allow us to draw all edges in $E(G) \cap E(G')$ i.e., those labeled as ordinary edges, in a planar way.
	Furthermore, we can draw these edges completely inside the frame, which ensures that no edge crosses edges from the frame.
	The remaining edges connect vertices of $G$ with the frame to model arching edges.
	Let $i \in [h]$ be a level that contains at least one arching edge in $\Gamma$.
	Recall that $G'$ contains edges of the form $ul_{i-1}$, $ul_{i}$, and $vr_i$, for each arching edge (arc) $uv \in A$, respectively.
	As $\Gamma$ is an arched level planar embedding that induces $\mathcal{L}$, all arching edges on level $i$ are incident to the same leftmost vertex in $\prec_i$, which is $u$. %
	Hence, $u$ is also the leftmost vertex on level $i$ in $\Gamma'$, which means that we can add the respective connections to the frame without introducing a crossing.
	Furthermore, for all other vertices $v$ that are endpoints of an arching edge on level $i$ are in $\Gamma$ we have $s_i \preceq_i v$, where $s_i$ is the rightmost vertex on level $i$ that has a neighbor on level $i+1$.
	Therefore, there is nothing that prevents us from inserting the edges between the frame and $r_i$ without introducing a crossing; see the blue edges in \Cref{fig:single-exp-ql-reduction}b.
	We conclude that $\Gamma'$ is a level planar drawing of $(G', \gamma')$.
	
	\proofsubparagraph*{($\boldsymbol{\Leftarrow}$)}
	Let $\Gamma'$ be a level planar drawing of $(G', \gamma')$.
	For the reasons given at the beginning of the proof, we focus only on the case where some edge in $G$ must arch according to $\mathcal{L}$.
	Therefore, there is at least one level $i \in [h]$ on which some vertices are connected to the frame.
	As $G$ and $G'$ are connected and the frame uses levels above and below the vertices of $V(G)$, the vertices of $G$ must be drawn inside the frame as in \Cref{fig:single-exp-ql-reduction}b.
	Recall that we have $\gamma'(v) = \gamma_{\mathcal{L}}(v)$ for every vertex $v \in V(G)$.
	Thus, we can for every level $i \in [h]$ use the vertex order of $\Gamma'$ to construct $\Gamma$ and it only remains to show that this order admits an arched level planar embedding that induces $\mathcal{L}$.
	Observe that all non-arching edges induce a level planar drawing $(G, \mathcal{L})$.
	Recall that for the arching edges, their endpoints are attached to the left and right side of the frame.
	Let $u$ be a vertex on level $i$ attached to the left side of the frame, i.e., an endpoint of an arc $uv \in A$ with $\delta(uv) = \LabelArch$.
	The triangle $u,l_i, l_{i-1}$ together with the fact that $G$ and $G'$ are connected ensures that $u$ is the leftmost vertex on level $i$ in $\Gamma'$ (and thus in $\Gamma$).
	Similarly, we have $s_i \preceq_i v$ for all vertices $v$ incident to $r_i$, where $s_i$ is again the rightmost vertex on level $i$ with a neighbor on level $i+1$.
	To obtain an arched leveled planar drawing that induces $\mathcal{L}$, we can now route the arching edges along the lower tip of the frame to ensure that they are crossing-free.
	Furthermore, all arching edges on the same level are incident to the same (leftmost) vertex; otherwise, we would have returned a trivial no instance when constructing $(G', \gamma'$).
	We conclude %
	that $\Gamma$ is %
	an arched level planar embedding of $G$ that induces $\mathcal{L}$.
\end{prooflater}
\end{statelater}

\thmsingleexp*
\begin{prooflater}{pthmsingleexp}
	Let $G$ be a graph on $n$ vertices.
	We can assume without loss of generality that $G$ is connected, otherwise, we apply the following arguments for each connected component of $G$.
	First, we check whether $G$ has at most $2n-3$ edges.
	If not, we can immediately reject thanks to a bound established by Heath and Rosenberg~\cite{heath1992laying}, see also Dujmovi{\'c} and Wood~\cite{dujmovic2004linear}.
	
	We branch on the $2^{\BigO{n}}$ possible labelings $\mathcal{L}=(A, \delta\colon A \to \LabelImage)$ of $G$.
	We use \Cref{lem:single-exp-ql-edge-labeling} in each branch to compute a level assignment $\gamma_{\mathcal{L}}$ of $V$ that is consistent with $\mathcal{L}$ or conclude that no such assignment exist.
	In the latter case, we reject the branch.
	In the former case, we can reduce the instance $(G, \mathcal{L})$ of \probname{Labeled Arched Level Planarity} to an instance $(G', \gamma')$ of \probname{Level Planarity}.
	We can solve the latter instance in linear time~\cite{JL.LPE.2002}.
	If it admits a level planar drawing $\Gamma'$, we convert it into an arched level planar embedding $\Gamma$ that induces $\mathcal{L}$ using the arguments behind \Cref{lem:single-exp-ql-edge-labeling,lem:single-exp-ql-reduction-level-planarity}.
	The spine order $\prec$ is then obtained by concatenating the individual (inverse) total orders in increasing level order as described by Heath and Rosenberg~\cite{heath1992laying}.
	If no branch admits the desired level planar drawing, we report that there does not exist a 1-page queue layout of $G$.
	This takes \BigO{n} time per branch, giving us an overall running time of $2^{\BigO{n}}$.
	
	It remains to show correctness of the algorithm.
	For the forward direction, assume that $G$ has a 1-page queue layout \queueL.
	Due to the already mentioned equivalence to arched leveled planar drawings, this means that $G$ has also such an embedding $\Gamma$, which induces a labeling $\mathcal{L}'$.
	Hence, we can compute in the respective branch a consistent level assignment $\gamma_{\mathcal{L}'}$ and by \Cref{lem:single-exp-ql-edge-labeling} there exists a (not necessarily identical) embedding $\Gamma'$ of $G$ that uses $\gamma_{\mathcal{L}'}$.
	Applying \Cref{lem:single-exp-ql-reduction-level-planarity}, we conclude that the constructed instance of \probname{Level Planarity} admits a drawing, i.e., we correctly report that $G$ admits a 1-page queue layout.	
	For the backward direction, we assume that one branch has a level planar drawing for the constructed instance $(G',\gamma')$.
	\Cref{lem:single-exp-ql-reduction-level-planarity} allows us to convert it into an arched level planar embedding of $G$ that induces $\mathcal{L}$.
	This witness of an arched leveled planar embedding of $G$ is sufficient to conclude that $G$ has a 1-page queue layout.
\end{prooflater}

\section{Concluding Remarks}
While our results improve the state of the art on computing linear layouts in several directions, they also highlight the prominent open questions in this area. In particular, Theorem~\ref{thm:vifpt} moves us closer to settling the long-standing open questions of whether treewidth or treedepth can be used to facilitate the computation of linear layouts~\cite{dujmovic2011book,bhore2020parameterized,ganian2021parameterized,bhore2022parameterized,GanianMOPR24}. At the same time, Theorems~\ref{thm:pagewidth} and~\ref{thm:singleexp} yield the question of whether these classical problems can be solved in single-exponential time.

\bibliography{references}

\ifthenelse{\boolean{cameraready}}{}{
\onlyShort{
	\appendix
	\newpage
    We use \stackL[G] and \queueL[G] to differentiate between stack and queue layouts, respectively. 
    
	\section{Omitted Details from Section~\ref{sec:fpt-vi}}
	\label{app:fpt-vi}

        \observationRamseySizeEquivClass*
        \label{obs:size-equiv-class*}
        \pobservationRamseySizeEquivClass

        \lemmaReducedComputation*
        \label{lem:reducedcomputation*}
        \plemmaReducedComputation

        \claimRamseyEdgeColors*
        \label{claim:edge-colors*}
        \pclaimRamseyEdgeColors

        \viBlockObservations

        \lemmaStructure*
        \label{lem:structure*}
        \plemmaStructure

        \lemmaLayoutSupergraph*
        \label{lem:layout-supergraph*}
        \proofLayoutSupergraph

	\newpage
	\section{Omitted Details from Section~\ref{sec:poly-pw}}
	\label{app:poly-pw}

  \stategraphFullDetails

  \subsection{Omitted Proofs}
	
	\lemmaPolyPwCuts*
	\label{lem:poly-pw-cuts*}
	\plemmaPolyPwCuts
			
	\lemmaPolyPwCutSide*
	\label{lem:poly-pw-cut-side*}
	\plemmaPolyPwCutSide
	
	\lemmaPolyPwCutInduce*
	\label{lem:poly-pw-cut-inudce-unique*}
	\plemmaPolyPwCutInduce
	
	\lemmaPolyPwFeasibleStates*
	\label{lem:poly-pw-feasible-states*}
	\plemmaPolyPwFeasibleStates
			
	\thmpagewidth*
	\label{thm:pagewidth*}
	\pthmpagewidth
	
	\newpage
	\section{Omitted Details from Section~\ref{sec:single-exp-ql}}
	\label{app:sec:single-exp-ql}
	
	\lemmaSingleExpEdgeLabeling*
	\label{lem:single-exp-ql-edge-labeling*}
	\plemmaSingleExpEdgeLabeling

    \subparagraph*{Reducing from labeled \probname{Arched Leveled Planarity} to \probname{Level Planarity}.}
    \constArchLevelPlan
	\plemmaSingleExpReductionLP
	
	\thmsingleexp*
	\label{thm:singleexp*}
	\pthmsingleexp

}
}

\end{document}